\documentclass[lettersize,journal]{IEEEtran}
\usepackage{algorithm}
\usepackage{algpseudocode}
\usepackage{float}
\usepackage{cite}
\usepackage{amsmath,amsfonts,amsthm,mathtools}
\usepackage{graphicx}
\usepackage{textcomp}
\usepackage{xcolor}
\usepackage{epstopdf}
\usepackage{multirow}
\usepackage{booktabs}
\usepackage{xurl}
\usepackage{array}
\usepackage{ragged2e}
\newcolumntype{L}[1]{>{\RaggedRight\hspace{0pt}}p{#1}}
\newtheorem{theorem}{Theorem}
\newtheorem{lemma}{Lemma}

\begin{document}

% Specific conference options
%\renewcommand{\figurename}{Fig.}
%\renewcommand{\tablename}{TABLE}

\title{On the Stability of Strategic Energy Storage Operation in Wholesale Electricity Markets (Extended Version)}

\author{Aviad Navon, Juri Belikov \IEEEmembership{Senior Member, IEEE}, Ariel Orda \IEEEmembership{Fellow, IEEE}, Yoash Levron \IEEEmembership{Senior Member, IEEE}

\thanks{The work of A. Navon was partially supported by the Israel Department of Energy (PhD fellowship). The work of J. Belikov was partially supported by the Estonian Research Council (grant No. PRG1463), and by the Estonian Centre of Excellence in Energy Efficiency, ENER (grant TK230) funded by the Estonian Ministry of Education and Research.
The work of Y. Levron was partially supported by the Israel Science Foundation. 
A. Navon, A. Orda and Y. Levron are with The Andrew and Erna Viterbi Faculty of Electrical and Computer Engineering, Technion---Israel Institute of Technology, Israel (e-mails: aviad.nav@gmail.com, ariel@ee.technion.ac.il and yoashl@ee.technion.ac.il correspondingly).
J. Belikov is with The Department of Software Science, Tallinn University of Technology, Estonia (email: juri.belikov@taltech.ee).}}

% OrcIds:
% Aviad Navon: 0000-0002-5885-8304
% Ariel Orda: 0000-0001-5561-6599
% Yoash Levron: 0000-0001-9775-1406
% Juri Belikov: 0000-0002-8243-7374

% \thanks{\fontfamily{phv}\selectfont{Paper submitted to the International Conference on Power Systems Transients (IPST2019) in Perpignan, France June 17-20, 2019}.}



\maketitle
\begin{abstract}

High shares of variable renewable energy necessitate substantial energy storage capacity. However, it remains unclear how to design a market that, on the one hand, ensures a stable and sufficient income for storage firms, and, on the other hand, maintains stable and affordable electricity costs for the consumers. Here, we use a game theoretic model to study storage competition in wholesale electricity markets. A main result is that these types of games are not necessarily stable. In particular, we find that under certain conditions, which imply a combination of a high share of variable renewable energy sources and low flexibility of conventional power plants, the system will not converge to an equilibrium. However, we demonstrate that a price cap on storage price bids can ensure convergence to a stable solution. 
Moreover, we find that when the flexibility of conventional power plants is low, while the storage usage for energy balancing increases with renewable energy generation, the profitability of using storage for the sole purpose of energy arbitrage decreases.

\end{abstract}


\section{Introduction}

\IEEEPARstart{T}{he} increasing adoption of variable renewable energy sources (VRES) creates a growing need for grid flexibility \cite{IEA2022}. 
Traditionally, fossil fuel power plants provided this flexibility, but as the share of VRES is increasing, power plants are nearing their operational limits. Energy storage is perhaps the most discussed form of grid flexibility, and, indeed, the global capacity of energy storage is rapidly increasing \cite{Katz2020}. For example, the authors of \cite{Holzinger2021} forecast that the global market size of stationary energy storage will grow from approximately \$30B and 50 GWh in 2023 to \$110B and 220 GWh in 2035.

As the capacity of VRES and energy storage are increasing, the operation of energy storage in electricity markets is becoming more complex. While currently energy storage firms often operate as price-takers, as storage capacity grows they are expected to become price-makers and have a stronger impact on the market operation and social welfare \cite{Cruise2018}. 
% Consequently, energy storage firms have a growing involvement in wholesale electricity markets. 
Another impact of the increase in storage capacity is market saturation and decline in profitability, giving rise to competition \cite{ZHAO2022}. Indeed, profit declines in energy storage services are already occurring in electricity markets worldwide, such as in California \cite{Sacklers2021} and Germany \cite{Figgener2021}. Thus, while energy storage has many potential value streams, % The realization of the potential value from storage highly depends on the market structure \cite{Cruise2018}. However, the role of ESS in electricity markets is not well understood. 
it remains unclear how to set market design and regulation that balances between the profit of storage, which is necessary to incentivize further investments, and social welfare, which is affected by energy storage operation through the cost of electricity, price stability, and greenhouse emissions. As a result, policy incentives remain few, inhibiting further development and adoption of energy storage \cite{Ferguson2018, Holzinger2021}. 

%% Literature survey
A large body of literature studies the optimal operation of energy storage systems (ESS). While most papers study the operation of a single ESS \cite{Machlev2020, Abdulla2018, Jin2018}, several papers analyze the optimal operation of multiple ESS and how they affect one another through the price of electricity \cite{Schill2010, Wang2014, Yaagoubi2015, Cruise2018, ZHAO2022, Zheng2023, Smets2023, Gu2024}. The latter is often done with game theory, which is a suitable tool for studying the strategic behavior of multiple energy storage firms \cite{Navon2020}. For example, the authors in \cite{Cruise2018} employ a Cournot model to study competition between multiple energy storage firms and show that their profit is proportional to the number of competing firms. A similar conclusion is reached in \cite{ZHAO2022}, who analyze a game between energy storage investors and show that as the number of investors increases, the energy storage capacity in the system increases, but the capacity and profit of each firm decreases. 

A few papers use game theory to study the strategic interaction between storage firms and the system operator. Modeling such interactions was recently identified as a critical modeling effort to support energy storage integration into wholesale electricity markets \cite{Zheng2023}. For example, in \cite{Huang2019} and \cite{Huang2021}, the authors model a Cournot game between energy storage firms and a system operator, showing that the participation of storage firms in the market always improves social welfare. Moreover, they find that as the number of firms increases to infinity, the social welfare converges to its maximal potential. Also, in \cite{Ojeda2022}, the authors employ a Stackelberg game between storage firms and a system operator and find that storage firms do not use their full storage capacity to increase the value of energy arbitrage, concluding that regulation is required to maximize social welfare. 

%% Gap - cost of alternatives to storage
The above-mentioned studies have an underlying assumption that storage competition in electricity markets is a stable process. However, they do not fully consider system limitations that may lead to instability. For example, several studies consider limitations of conventional power plants in their framework, such as \cite{Huang2019} and \cite{Huang2021}, but assume that VRES can be freely curtailed to avoid reaching the power plants' generation limitations. While this assumption is valid for a low share of VRES, it is not realistic for scenarios of high VRES, in which curtailment may become more expensive than storage \cite{MIRON2023, SOLOMON2019, Root2017, Cardenas2019}. Thus, the current literature does not adequately address the implications of system limitations on the strategic behavior of storage firms and on the {\em stability} of such games, which may impact market power, storage profitability, and price stability.

% Thus, the current literature does not adequately address the implications of system limitations on the strategic behavior of storage firms and the stability of such games, which may impact market power, storage profitability, and price stability.   
% Taking into account the potentially high cost of storage alternatives is crucial in understanding the strategic behavior of storage firms when operating around system limitations.  

In this study, we consider storage competition in wholesale electricity markets using a game theoretic model and show that these types of games are not necessarily stable. To this end, we formulate a game between a system operator and storage firms and consider system constraints that reflect limitations in the daily operation of both conventional power plants and ESS, e.g., conventional generation output, storage capacity, storage periodicity, and energy balance. 
%We also extend previous works by adding a periodicity constraint that ensures a daily ESS charge-discharge cycle, rendering the problem formulation more realistic.
% Moreover, we define a maximum for the storage price bids, which represents storage flexibility alternatives or an energy price cap (the lowest between them). 
We focus our analysis on scenarios with excess renewable energy generation, i.e., when the net load is lower than the minimal output of conventional power plants, such as when solar generation is high and demand is low. A main result is that a competition between storage firms in wholesale electricity markets is not necessarily stable, i.e., under certain conditions, the system will not converge to an equilibrium point. We provide analytical conditions for such instability. Nonetheless, we demonstrate that price stability can be obtained by limiting the maximal price bid (or the cost of alternative flexibility solutions). We also find that the profitability of energy arbitrage, for energy beyond the amount required for energy balance, decreases as VRES penetration increases. The main contributions of this paper are summarized
as follows.

\begin{itemize}
    \item We formulate a non-cooperative game model between a system operator and energy storage firms, taking into account system constraints associated with both conventional power plants and Energy Storage Systems (ESS). Our model considers limiting factors such as conventional generation output, storage capacity, storage periodicity, and energy balance.
    \item Building upon prior studies, we introduce a periodicity constraint to the model, ensuring a daily charge-discharge cycle for ESS. This addition enhances the realism of the problem formulation, addressing a gap in existing models that may not fully capture the practical aspects of storage operation.
    \item Our investigation reveals that competition among storage firms in wholesale electricity markets may exhibit instability under certain conditions. Analytical conditions are provided, indicating situations where the system fails to converge to an equilibrium point.
    \item Despite the identified instability, our paper demonstrates that price stability can be achieved through interventions such as imposing limits on maximal price bids or reducing the cost of alternative flexibility solutions. These findings suggest potential regulatory or market interventions to enhance overall system stability.
    \item We reveal that the profitability of energy arbitrage, beyond what is necessary for energy balance, diminishes with increasing penetration of Variable Renewable Energy Sources (VRES). This insight sheds light on the economic dynamics associated with the integration of renewable energy sources.
\end{itemize}

The rest of this article is organized as follows. In Section~\ref{sec:prob} the problem formulation is provided. Section~\ref{sec:analysis} provides a formal analysis of the game. Section~\ref{sec:sim_results} investigates a case study of California Independent System Operator (CAISO). Finally, Section~\ref{sec:conclusion} concludes the article.

\section{Problem Formulation}\label{sec:prob}

% This section presents the system model components and the optimization problems faced by a power system operator and energy storage firms. Then, we model the strategic behavior and interactions between the grid operator and storage firms through a non-cooperative game. 

\subsection{System Model}

We consider a power system that includes the following components:
\begin{itemize}
\item \textbf{Electric Load}. The electric load is a known discrete time signal $p_{D,k}$, which is defined over the time horizon $k = 0, \ldots, K - 1$, and represents the overall electric load.

\item \textbf{Variable Renewable Energy Power Plants}. The total generated power from variable renewable energy power plants, e.g., solar and wind, is a known discrete time signal, which is defined as $p_{RE,k}$, with $k = 0, \ldots, K - 1$.

\item \textbf{Net Load}. The net load is $p_{L,k}$ with $k = 0, \ldots, K - 1$ and represents the sum of the electric load and generation from VRES, i.e., $p_{L,k} = p_{D,k} + p_{RE,k}$, $\forall k \in \{0, \ldots, K - 1\}$. We denote the sum of load over the time horizon $k = 0, \ldots, K - 1$ as $p_{L,sum}$, i.e., $p_{L,sum} = \sum_{j = 0}^{K - 1} p_{L,j}$.
    
\item \textbf{Conventional Power Plants}. The total generated power from conventional power plants, e.g., coal, nuclear, and natural gas, is defined as $p_{g,k}$, with $k = 0, \ldots, K - 1$. It is assumed that this power output is bounded from below, such that $p_{g,k} \ge p_{g,\min }$. Throughout this paper, system flexibility is defined as the conventional power plants' minimal output, also known as the system's Reliability Must Run (RMR) \cite{Didsayabutra2022}. The generated power is associated with a cost $c_g(p_{g,k})$, which represents the aggregated cost of power generation from conventional power plants. We assume that this cost only depends on the momentarily generated power and that it is positive, monotonically increasing and strictly convex, to reflect the fact that more generated power usually requires the use of more expensive and less efficient generation units \cite{Kirschen2004}. 
    % Realistically, ${c_g}\left( {{p_{g,k}}} \right)$ may also depend on the time index, but here we ignore this dependency for simplicity.
    
\item \textbf{Energy Storage Systems}. The power flowing into the ESS $m$ is denoted as $p_{m,k}$, and the stored energy as $e_{m,k}$, such that
\begin{equation}
\begin{aligned}
e_{m,k} &= e_{m,0} + \Delta \cdot \left(\sum\limits_{j = 0}^{k - 1} p_{m,j} \right),\quad \text{for } k = 1, \ldots, K, \\
0 &\le e_{m,k} \le e_{m,\max}, \quad \text{for } k = 1, \ldots, K,
\end{aligned}
\end{equation}
where $e_{m,0}$ is the initial energy stored in the ESS of firm $m$, and $\Delta$ is a constant time interval that defines the time resolution of the electricity market.

Each storage system is also characterized by a cost
\begin{equation}
c_m(p_{m,k}) =
\begin{cases}
c_{chg,m}p_{m,k}, & p_{m,k} \ge 0\ \text{(charging)}, \\
c_{dis,m}p_{m,k}, & p_{m,k} < 0\ \text{(discharging)},
\end{cases}
\end{equation}
where $c_{chg,m}$ and $c_{dis,m}$ are the charging and discharging costs, respectively. Since the storage is operated to make a profit, we assume that $0 \le c_{chg,m} \le c_{dis,m}$.
    
We also require the solution to be periodical, to reflect a daily ESS (dis)charge cycle, so another constraint is $\sum_{j = 0}^{K - 1} p_{m,j}  = 0$.
\end{itemize}

\subsection{Power System Operator Optimization Problem}

The system operator dispatches the generators and storage systems to maintain a power balance in the system, such that
\begin{equation}
p_{g,k} = p_{L,k} + \sum_{m = 1}^M p_{m,k}, \ \text{for} \quad k = 0, \ldots, K - 1.
\end{equation}
    
The operator attempts to minimize the total cost of energy, hence it decides which combination of sources (generators or storage systems) to use at every moment, by solving the following optimization problem:

\begin{align}
&\min_{p_{g,k},p_{m,k}} \quad \sum_{k = 0}^{K - 1} \left(c_g(p_{g,k}) - \sum_{m = 1}^M c_m(p_{m,k}) \right) \notag \\
&\text{s.t.} \quad p_{g,k} = p_{L,k} + \sum_{m = 1}^M p_{m,k}, \ \text{for } k = 0, \ldots, K - 1, \label{eq:energyBalance_Cons} \\
&p_{g,k} \ge p_{g,\min}, \quad \text{for } k = 0, \ldots, K - 1, \label{eq:minConv_Cons}\\
&c_m(p_{m,k}) =\begin{cases}
c_{chg,m}p_{m,k}, & p_{m,k} \ge 0\ \text{(charging)}, \\
c_{dis,m}p_{m,k}, & p_{m,k} < 0\ \text{(discharging)},
\end{cases} \label{eq:storageCost_Cons} \\
&\qquad\qquad\qquad \text{for } m = 1, \ldots, M, \\
&e_{m,k} = e_{m,0} + \Delta \cdot\left(\sum_{j = 0}^{k - 1}p_{m,j}\right), \text{ for } k = 1, \ldots, K, \label{eq:SoC_Cons} \\
&0 \le e_{m,k} \le e_{m,\max } \text{ for } k = 1, \ldots, K,\label{eq:EScapacity_Cons} \\
&\sum_{j = 0}^{K - 1} p_{m,j}  = 0, \label{eq:periodicity_Cons}
\end{align}
where \eqref{eq:energyBalance_Cons} is the energy balance constraint, \eqref{eq:minConv_Cons} is the minimal conventional generation constraint, \eqref{eq:storageCost_Cons} is the cost of storage, \eqref{eq:SoC_Cons} is the energy storage state of charge, \eqref{eq:EScapacity_Cons} is the energy storage capacity constraint and \eqref{eq:periodicity_Cons} is the periodicity constraint.
  
\subsection{Reformulation of the Power System Operator Optimization Problem}
\label{subsec:reform_opt}

We proceeded to reformulate the grid operator's optimization problem in order to render it numerically solvable, through the following steps.

\begin{enumerate}
    \item Define new signals ${p_{chg,m,k}},{p_{dis,m,k}}$ such that 
    \begin{align}
       \begin{array}{l}
    {p_{m,k}} = {p_{chg,m,k}} - {p_{dis,m,k}},\;\;{\rm{for}}\;\;\;k = 0 \ldots \left( {K - 1} \right),\\
    {p_{chg,m,k}} \ge 0,\;\;\;{p_{dis,m,k}} \ge 0,\;\;{\rm{for}}\;\;\;k = 0 \ldots \left( {K - 1} \right),
    \end{array} 
    \end{align}
    where ${p_{chg,m,k}}$ and ${p_{dis,m,k}}$ represent the charging and discharging power of storage system $m$ at time interval $k$ correspondingly. We assume that these signals cannot be both positive, meaning that
    \begin{align}
     \begin{array}{l}
    {\forall} m,k,{\rm{ }}\;\;{\rm{if}}\;{p_{chg,m,k}} > 0\;\;{\rm{then}}\;{p_{dis,m,k}} = 0,\\
    {\forall}m,k,{\rm{ }}\;\;{\rm{if}}\;{p_{dis,m,k}} > 0\;\;{\rm{then}}\;{p_{chg,m,k}} = 0.
    \end{array}   
    \end{align}

    \item Define a constant $c_{m}$ that represents the difference between the cost of storage charging and discharging, i.e., ${c_{m}} = {c_{dis,m}} - {c_{chg,m}}$.

    \item Define vectors
 \begin{align}
 \begin{array}{l}
{x} = \left( {\begin{array}{*{20}{c}}
P\\
{{P_{chg}}}
\end{array}} \right),\\
{P} = {\left( \begin{array}{l}
{p_{1,0}}, \ldots ,{p_{M,0}},{p_{1,1}}, \ldots ,{p_{M,1}},\\
 \ldots ,{p_{1,\left( {K - 1} \right)}}, \ldots ,{p_{M,\left( {K - 1} \right)}}
\end{array} \right)^{\rm{T}}},\\
{{{P_{chg}}}} = \\{\left( \begin{array}{l}
{p_{chg,}}_{1,0}, \ldots ,{p_{chg,M,0}},{p_{chg,1,1}}, \ldots ,{p_{chg,M,1}},\\
 \ldots ,{p_{chg,1,\left( {K - 1} \right)}}, \ldots ,{p_{chg,M,\left( {K - 1} \right)}}
\end{array} \right)^{\rm{T}}}.
\end{array}       
    \end{align}

\end{enumerate}

Based on the steps above, the problem can be rewritten as
\begin{align}
\begin{array}{l}
\mathop {\min }\limits_x \;\;\;\;\sum\limits_{k = 0}^{K - 1} {{c_g}\left( {{p_{L,k}} + \sum\limits_{m = 1}^M {{p_{m,k}}} } \right)}  + \sum\limits_{m = 1}^M {{c_{m}}\left( {\sum\limits_{k = 0}^{K - 1} {{p_{chg,m,k}}} } \right)} \\
{\rm{s}}{\rm{.t}}{\rm{.}}\;\;\;\;\;\;\sum\limits_{m = 1}^M {{p_{m,k}}}  \ge {p_{g,\min }} - {p_{L,k}}\;\;\;,\;\;\;\;{\rm{for}}\;\;\;k = 0 \ldots \left( {K - 1} \right)\\
\;\;\;\;\;\;\;\;\;\;{\rm{for}}\;\;\;m = 1 \ldots M:\\
\;\;\;\;\;\;\;\;\;\;{p_{chg,m,k}} \ge 0\;\;\;\;\;\;\;\;\;\;\;\;,\;\;\;\;{\rm{for}}\;\;\;k = 0 \ldots \left( {K - 1} \right)\\
\;\;\;\;\;\;\;\;\;\;{p_{m,k}} - {p_{chg,m,k}} \le 0\;\;\;,\;\;\;\;{\rm{for}}\;\;\;k = 0 \ldots \left( {K - 1} \right)\\
\;\;\;\;\;\;\;\;\;\;{e_{m,k}} = {e_{m,0}} + \Delta  \cdot \left( {\sum\limits_{j = 0}^{k - 1} {{p_{m,j}}} } \right),\;{\rm{for}}\;k = 1 \ldots \left( {K - 1} \right)\\
\;\;\;\;\;\;\;\;\;\;0 \le {e_{m,k}} \le {e_{m,\max }}\;\;\;\;\;,\;\;\;\;{\rm{for}}\;\;\;k = 1 \ldots \left( {K - 1} \right)\\
\;\;\;\;\;\;\;\;\;\;\sum\limits_{j = 0}^{K - 1} {{p_{m,j}}}  = 0.
\end{array}    
\end{align}

To the above reformulation of the optimization problem we add the assumption that the cost ${c_g}\left(  \cdot  \right)$ is quadratic, and is given by
\begin{align}
 \label{eq:quad_cost_f}
    {c_g}\left( x \right) = \frac{1}{2}a{x^2} + b.   
\end{align}

As explained in \cite{Rau2003, Marijana2012}, this is a common and reasonable assumption for the cost function of thermal power plants. Then, this problem can be solved as a quadratic program, which is formulated as
\begin{align}
    \begin{array}{l}
\mathop {\min }\limits_x \;a\sum\limits_{k = 0}^{K - 1} {\left( {\frac{1}{2}{p_{L,k}}^2 + {p_{L,k}}\left( {\sum\limits_{m = 1}^M {{p_{m,k}}} } \right) + \frac{1}{2}{{\left( {\sum\limits_{m = 1}^M {{p_{m,k}}} } \right)}^2}} \right)} \\
\;\;\;\;\;\;\;\;\;\; + \sum\limits_{k = 0}^{K - 1} {\frac{b}{a} \cdot \left( {{p_{L,k}} + \sum\limits_{m = 1}^M {{p_{m,k}}} } \right)} \\
\,\,\,\,\,\,\,\,\,\,\,\,\,\,\, + \sum\limits_{m = 1}^M {\frac{{{c_{m}}}}{a}\left( {\sum\limits_{k = 0}^{K - 1} {{p_{chg,m,k}}} } \right)} \\
{\rm{s}}{\rm{.t}}{\rm{.}}\;\;\;\;\;\;\sum\limits_{m = 1}^M {{p_{m,k}}}  \ge {p_{g,\min }} - {p_{L,k}}\;\;\;\;\;\;,\;{\rm{for}}\;\;\;k = 0 \ldots \left( {K - 1} \right)\\
\;\;\;\;\;\;\;\;\;\;{\rm{for}}\;\;\;m = 1 \ldots M:\\
\;\;\;\;\;\;\;\;\;\;{p_{chg,m,k}} \ge 0\;\;\;\;\;\;\;\;\;\;\;\;\;\;\;\;\;\;\;\;,\;{\rm{for}}\;\;\;k = 0 \ldots \left( {K - 1} \right)\\
\;\;\;\;\;\;\;\;\;\;{p_{m,k}} - {p_{chg,m,k}} \le 0\;\;\;\;\;\;\;\;\;\;\;,\;{\rm{for}}\;\;\;k = 0 \ldots \left( {K - 1} \right)\\
\;\;\;\;\;\;\;\;\;\;{e_{m,k}} = {e_{m,0}} + \Delta  \cdot \left( {\sum\limits_{j = 0}^{k - 1} {{p_{m,j}}} } \right),\;{\rm{for}}\;k = 1 \ldots \left( {K - 1} \right)\\
\;\;\;\;\;\;\;\;\;\;0 \le {e_{m,k}} \le {e_{m,\max }}\;\;\;\;\;\;\;\;\;\;\;\;\;,\;{\rm{for}}\;\;\;k = 1 \ldots \left( {K - 1} \right)\\
\;\;\;\;\;\;\;\;\;\;\sum\limits_{j = 0}^{K - 1} {{p_{m,j}}}  = 0.
\end{array}
\end{align}

Now, since $a$, $\frac{1}{2}a{p_{L,k}}^2$ and $\sum\limits_{k = 0}^{K - 1} b {p_{L,k}}$ are constants and $\sum\limits_{k = 0}^{K - 1} {\sum\limits_{m = 1}^M {{p_{m,k}}}  = 0}$, we get the equivalent problem
\begin{align} \label{eq:ISO_opt}
    \begin{array}{l}
\mathop {\min }\limits_x\;\;\;\;\;\frac{1}{2}\sum\limits_{k = 0}^{K - 1} {{{\left( {\sum\limits_{m = 1}^M {{p_{m,k}}} } \right)}^2}}  + \sum\limits_{k = 0}^{K - 1} {\sum\limits_{m = 1}^M {{p_{L,k}}{p_{m,k}}} } \\
\,\,\,\,\,\,\,\,\,\,\,\,\,\, + \sum\limits_{k = 0}^{K - 1} {\sum\limits_{m = 1}^M {\frac{{{c_{m}}}}{a}{p_{chg,m,k}}} } \\
{\rm{s}}{\rm{.t}}{\rm{.}}\;\;\;\\
\, - \sum\limits_{m = 1}^M {{p_{m,k}}}  \le {p_{L,k}} - {p_{g,\min }},\;{\rm{for}}\;k = 0 \ldots \left( {K - 1} \right)\;\\
{\rm{for}}\;\;\;m = 1 \ldots M:\;\;\;\\
\,{p_{m,k}} - {p_{chg,m,k}} \le 0\;,\;{\rm{for}}\;\;k = 0 \ldots \left( {K - 1} \right)\;\\
\, - {p_{chg,m,k}} \le 0\;\;\;\;\;\;\;\;,\;{\rm{for}}\;\;k = 0 \ldots \left( {K - 1} \right)\;\\
\;\,\sum\limits_{j = 0}^{k - 1} {{p_{m,j}}}  \le \frac{{{e_{m,\max }} - {e_{m,0}}}}{\Delta },\;\;{\rm{for}}\;k = 1 \ldots \left( {K - 1} \right)\;\\
\, - \left( {\sum\limits_{j = 0}^{k - 1} {{p_{m,j}}} } \right) \le \; + \frac{{{e_{m,0}}}}{\Delta },\;\;f{\rm{or}}\;k = 1 \ldots \left( {K - 1} \right)\;\\
\,\sum\limits_{j = 0}^{K - 1} {{p_{m,j}}}  = 0.
\end{array}
\end{align}

This problem can be written in matrix form as

\begin{align}\begin{array}{l}
\min ,\;\;\;\;\;\frac{1}{2}{x^{\rm{T}}}Qx + {r^{\rm{T}}}x\\
{\rm{s}}{\rm{.t}}{\rm{.}}\;\;\;\;\;\;\;\left( {\begin{array}{*{20}{c}}
{{H_1}}\\
{{H_2}}\\
{[{H_3}\;\;{0_{MK \times MK}}]}\\
{[ - {H_3}\;\;{0_{MK \times MK}}]}\\
{[ - {H_4}\;\;{0_{K \times MK}}]}
\end{array}} \right)x \le \left( {\begin{array}{*{20}{c}}
{{0_{MK \times 1}}}\\
{{0_{MK \times 1}}}\\
{{g_2}}\\
{{g_3}}\\
{{g_4}}
\end{array}} \right),\\
\left[ {{H_5}\;\;\;\;{0_{M \times MK}}} \right]x = {0_{M \times 1}},
\end{array}\end{align}
where
\begin{align}\begin{array}{l}
{x_{2MK \times 1}} = \left( {\begin{array}{*{20}{c}}
P\\
{{P_{chg}}}
\end{array}} \right)\\
{\left( P \right)_{MK \times 1}} = {\left( \begin{array}{l}
{p_{1,0}}, \ldots ,{p_{M,0}},{p_{1,1}}, \ldots ,{p_{M,1}},\\
 \ldots ,{p_{1,\left( {K - 1} \right)}}, \ldots ,{p_{M,\left( {K - 1} \right)}}
\end{array} \right)^{\rm{T}}}\\
{\left( {{P_{chg}}} \right)_{MK \times 1}} = \\ {\left( \begin{array}{l}
{p_{chg,}}_{1,0}, \ldots ,{p_{chg,M,0}},{p_{chg,1,1}}, \ldots ,{p_{chg,M,1}},\\
 \ldots ,{p_{chg,1,\left( {K - 1} \right)}}, \ldots ,{p_{chg,M,\left( {K - 1} \right)}}
\end{array} \right)^{\rm{T}}}\\
{\left( {{Q_a}} \right)_{MK \times MK}} = \left( {\begin{array}{*{20}{c}}
{{1_{M \times M}}}&{}&{{0_{M \times M}}}\\
{}& \ddots &{}\\
{{0_{M \times M}}}&{}&{{1_{M \times M}}}
\end{array}} \right)\\
{Q_{2MK \times 2MK}} = \left( {\begin{array}{*{20}{c}}
{{Q_a}}&{{0_{MK \times MK}}}\\
{{0_{MK \times MK}}}&{{0_{MK \times MK}}}
\end{array}} \right) \\ + {\varepsilon _0}{I_{2MK \times 2MK}}
\end{array}\end{align}

where ${\varepsilon _0}$ is very small, and

\begin{align}\begin{array}{l}
{\left( {{r_a}} \right)_{MK \times 1}} = {\left( {\overbrace {{p_{L,0}}, \ldots ,{p_{L,0}}}^{M\;{\rm{times}}}, \ldots ,\overbrace {{p_{L,K - 1}}, \ldots ,{p_{L,K - 1}}}^{M\;{\rm{times}}}} \right)^{\rm{T}}}\\
{\left( {{r_b}} \right)_{MK \times 1}} = \\ {\left( {\overbrace {\left[ {\frac{{{c_{1}}}}{a}, \ldots ,\frac{{{c_{M}}}}{a}} \right], \ldots ,\left[ {\frac{{{c_{1}}}}{a}, \ldots ,\frac{{{c_{M}}}}{a}} \right]}^{K\;{\rm{times}}}} \right)^{\rm{T}}}\\
{r_{2MK \times 1}} = \left( {\begin{array}{*{20}{c}}
{{r_a}}\\
{{r_b}}
\end{array}} \right)\\
{\left( {{H_1}} \right)_{MK \times 2MK}} = \left( {{I_{MK \times MK}}\;\;,\;\; - {I_{MK \times MK}}} \right)\\
{\left( {{H_2}} \right)_{MK \times 2MK}} = \left( {{0_{MK \times MK}}\;\;,\;\; - {I_{MK \times MK}}} \right)\\
{\left( {{H_3}} \right)_{MK \times MK}} = \left( {\begin{array}{*{20}{c}}
{{I_{M \times M}}}&{}&{}&{}\\
{{I_{M \times M}}}&{{I_{M \times M}}}&{}&{}\\
{{I_{M \times M}}}&{{I_{M \times M}}}&{{I_{M \times M}}}&{}\\
 \vdots &{}& \vdots & \ddots 
\end{array}} \right)\\
{\left( {{H_4}} \right)_{K \times MK}} = \left( {\begin{array}{*{20}{c}}
{{1_{1 \times M}}}&{{0_{1 \times M}}}& \cdots &{{0_{1 \times M}}}\\
{{0_{1 \times M}}}&{{1_{1 \times M}}}&{{0_{1 \times M}}}& \cdots \\
 \vdots &{}& \ddots & \vdots \\
{{0_{1 \times M}}}& \cdots &{{0_{1 \times M}}}&{{1_{1 \times M}}}
\end{array}} \right)\\
{\left( {{H_5}} \right)_{M \times MK}} = \left( {\begin{array}{*{20}{c}}
{{I_{M \times M}}}&{{I_{M \times M}}}& \cdots &{{I_{M \times M}}}
\end{array}} \right)\\
{\left( {{g_2}} \right)_{MK \times 1}} = \frac{1}{\Delta }\left( {\begin{array}{*{20}{c}}
{{e_{1,\max }} - {e_{1,0}}}\\
 \vdots \\
{{e_{M,\max }} - {e_{M,0}}}\\
 \vdots \\
{{e_{1,\max }} - {e_{1,0}}}\\
 \vdots \\
{{e_{M,\max }} - {e_{M,0}}}
\end{array}} \right)\;\;\\
{\left( {{g_3}} \right)_{MK \times 1}} = \frac{1}{\Delta }\left( {\begin{array}{*{20}{c}}
{{e_{1,0}}}\\
 \vdots \\
{{e_{M,0}}}\\
 \vdots \\
{{e_{1,0}}}\\
 \vdots \\
{{e_{M,0}}}
\end{array}} \right)\\
{\left( {{g_4}} \right)_{K \times 1}} = \left( {\begin{array}{*{20}{c}}
{{p_{L,0}} - {p_{g,\min }}}\\
 \vdots \\
{{p_{L,K - 1}} - {p_{g,\min }}}
\end{array}} \right)
\end{array}.\end{align}

\subsection{Energy Storage Firms' Optimization Problem}

Each energy storage system $m$ is owned by a storage firm that participates in a wholesale electricity market to gain profit from energy arbitrage. Each storage firm offers the system operator a price bid ${c_{m}}$ for using its storage system, i.e., the cost difference between charging and discharging, and aims to maximize its profit by solving the problem
\begin{equation}
\label{eq:storage_opt}
\begin{aligned}
\max_{c_m} &&& c_m\sum_{k = 0}^{K - 1} p_{chg,m,k} \\
\text{s.t.} &&& \delta  \le {c_m} \le c_{\max},
\end{aligned}
\end{equation}
where $\delta$ and $c_{\max}$ are the minimal and maximal price bids, respectively.

\subsection{Game Model}

We model a non-cooperative game with the following components:
\begin{itemize}
\item \textbf{Players}. Power system operator and $M$ energy storage firms.
\item  \textbf{Strategies}. Each energy storage firm $m$ decides on a price bid $c_{m}$ for using its storage system, i.e., the cost difference between charging and discharging, and the power system operator chooses $x$, i.e., the power flow to each storage system at each time index. 
    % Since the net load is a given then the energy balance constraint $(\ref{eq:energyBalance_Cons})$ can be used to compute the conventional generation at each time index, i.e., $p_{g,k}$.
\item \textbf{Utility functions}. The system operator aims to minimize electricity costs by solving optimization problem~\eqref{eq:ISO_opt}. The energy storage firms aim to maximize their profit by solving optimization problem~\eqref{eq:storage_opt}. The mutual influence between the storage firms and the power system operator is reflected in their utility functions: the storage firms' utility function depends on $\sum_{k = 0}^{K - 1} p_{chg,m,k}$, i.e., the power dispatch by the system operator, which depends on $c_{m} \text{for } m = 1, \ldots, M$, i.e., the firms' price bids, and vice versa. 
\end{itemize}

 % We consider best-response dynamics among the players \cite{Fudenberg1991}. A flow chart of the game-theoretic model is depicted in Fig.~\ref{fig:FlowChart}.  
 
We solve the game using a best-reply algorithm. That is, the players iteratively choose a strategy that minimizes their utility function and the algorithm stops if/when it reaches a Nash Equilibrium (NE), i.e., a solution from which no player can benefit by unilaterally deviating from it \cite{Fudenberg1991}. 
 % If an equilibrium exists, the model returns the price bids and power dispatch of storage at equilibrium, which can be used to compute the power dispatch of conventional power plants, system marginal price, and the storage firms' profit. 
 A flow chart of the game-theoretic model is depicted in Fig.~\ref{fig:FlowChart}.
 % The model inputs include a daily net-load profile (demand profile minus variable renewable energy generation), conventional generation cost function, conventional power output constraints, storage operation cost, storage capacity constraints, and a storage price bid cap. 

 \begin{figure*}[htbp]
     \centering
     \includegraphics[width=0.7\textwidth]{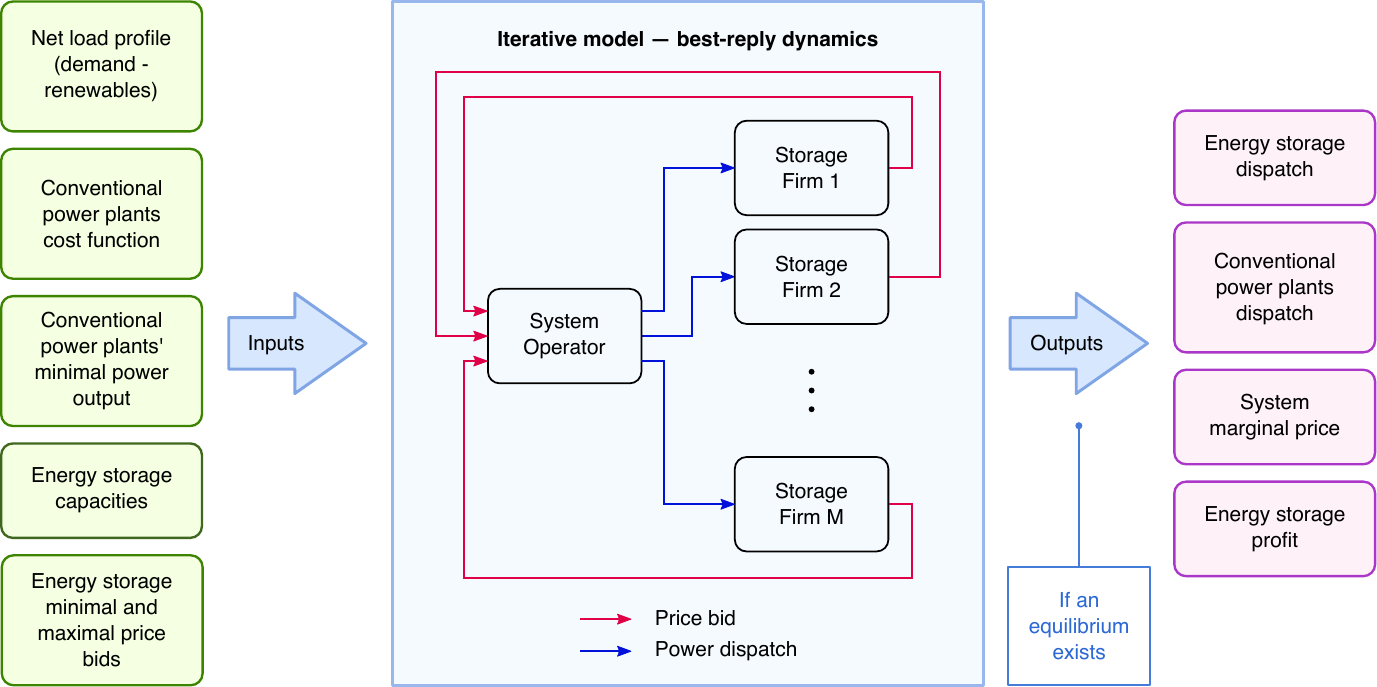}
     \caption{A flow chart of the game-theoretic model.}
     \label{fig:FlowChart}
 \end{figure*}

\section{Formal Analysis}\label{sec:analysis}

In this section, we provide analytical conditions for the instability of the game. In particular, we prove that under certain conditions, the best-response algorithm does not converge. We begin with a set of definitions and assumptions.

% In particular, we find that when the initial storage price offer is high and the storage capacity of the smallest storage firm is smaller than the required storage for energy balancing, then the storage firms that offer price bids do not reach an agreement, i.e., a NE. 

\subsection{Definitions}
\begin{itemize}
\item We define a set of time indices, as follows:
\begin{align}
K_{absorb} = \{k \in K \mid 0 \le p_{g,\min} - p_{L,k}\},
\end{align} 
and we denote the minimal and maximal indices in $K_{absorb}$ as $k_{absorb,1}$ and $k_{absorb,2}$, respectively.

\item We define a constant that reflects the excess generated energy in the system, as follows:
\begin{align}
E_{absorb} = \sum_{k \in K_{absorb}} \left(p_{g\min} - p_{L,k}\right).
\end{align}
    
\item We use the notation ``$M/m$'' to denote all storage firms except for storage firm $m$. 
\end{itemize}

\subsection{Assumptions}

\begin{itemize}
\item The net load's minimum is lower than the conventional power plants' minimal output, i.e., $p_{g,\min} > \min (p_{L,k}).$
\item The net load equals the conventional power plants' minimal output twice during the time horizon of the game, i.e., $|\{ k \in K \mid p_{L,k} = p_{g,\min}\}| = 2$.
\item The initial state of charge of all storage systems is zero, i.e., $e_{m,0} = 0$, $m = 1,\ldots,M$.
% \item The capacity of energy storage $1$ is either larger or equal to the capacity of energy storage $2$, i.e.,  $e_{\max ,1} = \max (e_{\max,1},e_{\max,2})$.
 \item ${c_{m}},\,m = 1 \ldots M$ are discrete, i.e., ${c_{m}} \in \{ \delta,2 \delta , \ldots ,{c_{\max }} - \delta ,{c_{\max }}\} $, where $ \delta \ll {c_{\max }} $.
%The value of ${c_{\min }}$ reflects the cost of energy storage operation and ${c_{\max }}$ is the maximal bidding price. 

\item If two price bids yield an identical profit for a storage firm, it will choose the lower bid, i.e., if $J_m({c^*}_{m}) = J_m({\tilde c_{m}})$ and ${c^*}_{m} < {\tilde c}_{m}$ then storage firm $m$ will offer price bid ${c^*}_{m}$.
\item The players follow best-response dynamics. We denote the action of storage firm $m$ in iteration $t$ of the best-response dynamics as $c_{m}^{(t)}$.
\end{itemize}

\subsection{Instability Theorem}

Next, in Lemmas~\ref{lem:minimal_charge}-\ref{lem:optimal_c_diff}, we incrementally develop the mathematical foundations that are necessary to prove Theorem~\ref{theo:infinite_loop}. Due to space limitations, we omit the proofs of Lemmas~\ref{lem:minimal_charge}-\ref{lem:optimal_c_diff} and direct the reader to (cite arXiv) for their details. 

\begin{lemma}
\label{lem:minimal_charge}
\begin{equation}
\left(p_{g,\min} - p_{L,k}\right)\le \sum_{m = 1}^M p_{m,k} \le \sum_{m = 1}^M p_{chg,m,k}, \ \forall k \in K
\end{equation}
\end{lemma}

\begin{proof}

Recall that the grid operator has an energy balance constraint, i.e., 

\[\;\;\;{p_{g,k}} = {p_{L,k}} + \sum\limits_{m = 1}^M {{p_{m,k}}} \;,\;\;\;\;{\rm{for}}\;\;\;k = 0 \ldots \left( {K - 1} \right)\]

Using ${p_{g,\min }} \le {p_{g,k}}\,{\rm{for}}\;k = 0 \ldots \left( {K - 1} \right)$ we get

\[{p_{g,\min }} - {p_{L,k}} \le {p_{g,k}} - {p_{L,k}} = \sum\limits_{m = 1}^M {{p_{m,k}}} ,\;{\rm{for}}\;k = 0 \ldots \left( {K - 1} \right)\]

We recall that the grid operator’s optimization problem has the constraint 

\[\begin{array}{l}
{\rm{for}}\;\;\;m = 1 \ldots M:\;\;\;\\
\,{p_{m,k}} - {p_{chg,m,k}} \le 0\;,\;{\rm{for}}\;\;k = 0 \ldots \left( {K - 1} \right)
\end{array}\]

Hence overall we get 

\[\left( {{p_{g,\min }} - {p_{L,k}}} \right) \le \sum\limits_{m = 1}^M {{p_{m,k}}}  \le \sum\limits_{m = 1}^M {{p_{chg,m,k}}}, {\rm{ }}\forall k \in K\]

\end{proof}

\begin{lemma}
\label{lem:f_bounded}
There exists a positive constant $W$ such that $-W \le f(P) \le W$, where
\begin{equation}
f(P) = \frac{1}{2}\sum_{k = 0}^{K - 1} \left(\sum_{m = 1}^M p_{m,k}\right)^2 + \sum_{k = 0}^{K - 1} \sum_{m = 1}^M p_{L,k}p_{m,k}.
\end{equation}
\end{lemma}

\begin{proof}
We first prove by induction on $k$ that for $k = 0 \ldots \left( {K - 1} \right), m = 1...M$, the variables ${p_{m,k}}$ are bounded, that is, there exists a compact set $S$ such that $P \in S$.

\textit{Induction proposition:} for every $n \in \mathbb{N}_+$, if \[0 \le \sum\limits_{j = 0}^{k - 1} {{p_{m,j}}}  \le \frac{{{e_{m,\max }}}}{\Delta }\,{\rm{for}}\,k = 1 \ldots \left( {n + 1} \right)\] then the variables ${p_{m,k}}$ where $k = 0 \ldots n$ are bounded. 

\textit{base case:} for $n=0$ if \[0 \le \sum\limits_{j = 0}^{k - 1} {{p_{m,j}}}  \le \frac{{{e_{m,\max }}}}{\Delta }\,{\rm{for}}\,k = 1\] then clearly ${p_{m,0}}$ is bounded.

\textit{Induction step:} Assume that for $k = t - 1$ if \[ - \frac{{{e_{m,0}}}}{\Delta }\, \le \sum\limits_{j = 0}^{k - 1} {{p_{m,j}}}  \le \frac{{{e_{m,\max }} - {e_{m,0}}}}{\Delta }\,{\rm{for}}\,k = 1 \ldots t\] then the variables ${p_{m,k}}, k = 0 \ldots \left( {t - 1} \right)$ are bounded.

If \[ - \frac{{{e_{m,0}}}}{\Delta }\, \le \sum\limits_{j = 0}^{k - 1} {{p_{m,j}}}  \le \frac{{{e_{m,\max }} - {e_{m,0}}}}{\Delta }\,{\rm{for}}\,k = 1 \ldots \left( {t + 1} \right)\] then based on our assumption the variables ${p_{m,k}}, k = 0 \ldots  {t} $ are bounded. In addition, by subtracting $\sum\limits_{j = 0}^{t - 1} {{p_{m,j}}}$ from both sides of $ - \frac{{{e_{m,0}}}}{\Delta }\, \le \sum\limits_{j = 0}^t {{p_{m,j}}}  \le \frac{{{e_{m,\max }} - {e_{m,0}}}}{\Delta }\,$ we get 

\[ - \frac{{{e_{m,0}}}}{\Delta }\, - \sum\limits_{j = 0}^{t - 1} {{p_{m,j}}}  \le {p_{m,t}} \le \frac{{{e_{m,\max }} - {e_{m,0}}}}{\Delta } - \sum\limits_{j = 0}^{t - 1} {{p_{m,j}}} \]

and since $\sum\limits_{j = 0}^{t - 1} {{p_{m,j}}}$ is bounded then the variable ${p_{m,t}}$ is bounded as well. 

Based on both the problem’s constraint \[\; - \frac{{{e_{m,0}}}}{\Delta }\, \le \sum\limits_{j = 0}^{k - 1} {{p_{m,j}}}  \le \frac{{{e_{m,\max }} - {e_{m,0}}}}{\Delta },\;\;{\rm{for}}\;k = 1 \ldots K,\;\] and the induction preposition, it follows that the variables ${p_{m,k}},\,k = 0 \ldots K - 1,$ are bounded, and more specifically, there exists a compact set $S$ such that $P \in S$.

Finally, since $f(P)$ is continuous and $P$ belongs to a compact set then there exists a positive constant $W$ such that $ - W \le f(P) \le W$ \cite{Protter1977}, where \[f(P) = \frac{1}{2}\sum\limits_{k = 0}^{K - 1} {{{\left( {\sum\limits_{m = 1}^M {{p_{m,k}}} } \right)}^2}}  + \sum\limits_{k = 0}^{K - 1} {\sum\limits_{m = 1}^M {{p_{L,k}}{p_{m,k}}} }.\]
\end{proof}

\begin{lemma}
\label{lem:optimal_charge_over_m}
Consider an optimal solution $x^*$. There exists a constant $c_{\min}$ for which if $c_{m} \ge c_{\min}\ \forall m$ then it holds for the optimal solution that
\begin{equation}
\sum_{m = 1}^M p_{chg,m,k}^*  = \max(0,p_{g,\min} - p_{L,k}),\ \forall k \in K.
\end{equation}
\end{lemma}

\begin{proof}
The original grid operator optimization problem can be viewed as a specific instance of the following problem:

\[\begin{array}{l}
\min .\;\;\;\;\;f\left( P \right)\, + h({P_{chg}})\\
{\rm{s}}{\rm{.t}}{\rm{.}}\;\\
{P_{chg}} \ge 0\\
P \in S\left( {{P_{chg}}} \right)
\end{array}\]

where

\[\begin{array}{l}
{\left( P \right)_{MK \times 1}} = {\left( \begin{array}{l}
{p_{1,0}}, \ldots ,{p_{M,0}},{p_{1,1}}, \ldots ,{p_{M,1}},\\
 \ldots ,{p_{1,\left( {K - 1} \right)}}, \ldots ,{p_{M,\left( {K - 1} \right)}}
\end{array} \right)^{\rm{T}}}\\
{\left( {{P_{chg}}} \right)_{MK \times 1}} = \\
{\left( \begin{array}{l}
{p_{chg,}}_{1,0}, \ldots ,{p_{chg,M,0}},{p_{chg,1,1}}, \ldots ,{p_{chg,M,1}},\\
 \ldots ,{p_{chg,1,\left( {K - 1} \right)}}, \ldots ,{p_{chg,M,\left( {K - 1} \right)}}
\end{array} \right)^{\rm{T}}}\\
f(P) = \frac{1}{2}\sum\limits_{k = 0}^{K - 1} {{{\left( {\sum\limits_{m = 1}^M {{p_{m,k}}} } \right)}^2}}  + \sum\limits_{k = 0}^{K - 1} {\sum\limits_{m = 1}^M {{p_{L,k}}{p_{m,k}}} } \\
h({P_{chg}}) = \sum\limits_{k = 0}^{K - 1} {\sum\limits_{m = 1}^M {{c_{m}}{p_{chg,m,k}}} } 
\end{array}\]

and $S$ is a convex set.

Consider an optimal solution ${x^*}$. To prove that

\[\sum\limits_{m = 1}^M {{p_{chg,m,k}}^*}  = \max (0,{p_{g,\min }} - {p_{L,k}})\,\forall k \in K,\]

we recall that 

\[\sum\limits_{m = 1}^M {{p_{chg,m,k}}^*}  \ge \max (0,{p_{g,\min }} - {p_{L,k}}), \forall k \in K,\]

(Lemma \ref{lem:minimal_charge}) and assume by negation that for the optimal solution ${x^*}$ it holds that $\sum\limits_{m = 1}^M {{p_{chg,m,k}}^*}  > \max (0,{p_{g,\min }} - {p_{L,k}})$ for some time indices $k \in \tilde K$, where $\tilde K \subseteq K$. Consequently, there exists a feasible solution $\tilde x$ that is slightly smaller than ${x^*}$ for some time indices $k \in \tilde K$, i.e., $\tilde x = {x^*} - td$ where $d \in {R^n},\,\,t \in R,\,\,t \to 0$.

Since ${x^*}$ is optimal, we can claim that if we change ${x^*}$ to $\tilde x$, then the objective function will increase, that is

\[f\left( {{P^*}} \right)\, + h({P_{chg}}^*) \le f\left( {\tilde P} \right)\, + h({\tilde P_{chg}}),\]

which can be more explicitly written as

\[\begin{array}{l}
f\left( {{P^*}} \right)\, + \sum\limits_{k = 0}^{K - 1} {\sum\limits_{m = 1}^M {{c_{m}}{p^*}_{chg,m,k}} } \\
 \le f\left( {{P^*} - td} \right)\, + \sum\limits_{k \notin \tilde K} {\sum\limits_{m = 1}^M {{c_{m}}{p^*}_{chg,m,k}} } \\
 + \sum\limits_{k \in \tilde K} {\sum\limits_{m = 1}^M {{c_{m}}\left( {{p^*}_{chg,m,k} - t} \right)} } 
\end{array}.\]

By rearranging and canceling out elements on both sides we get

\[\sum\limits_{k \in \tilde K} {\sum\limits_{m = 1}^M {{c_{m}}} }  \le \frac{{f\left( {{P^*} - td} \right) - f\left( {{P^*}} \right)\,\,}}{t}.\]

Since $f$ is quadratic then it is smooth \cite{Boyd2004}, i.e.,

\[\left\| {\nabla f({P_1}) - \nabla f({P_2})} \right\| \le 2\left\| {{P_1} - {P_2}} \right\|\,\forall {P_1},{P_2} \in S,\]

and since $S$ is a compact set then $\nabla f(P)$ is bounded.
Since $f$ is continuously differentiable over $S$ then 

\[\begin{array}{l}
\nabla f{(P)^T}d = f'(P;d) = \,\mathop {\lim }\limits_{t \to {0^ + }} \frac{{f(P + td) - f(P)}}{t},\\
\forall P \in S,d \in {R^n}.
\end{array}\]

It follows that $\mathop {\lim }\limits_{t \to {0^ + }} \frac{{f(P + td) - f(P)}}{t},\,P \in S,d \in {R^n}$ is bounded by ${W_2}$.

By taking $t$ to zero it follows from the inequality above that 

\[\sum\limits_{k \in \tilde K} {\sum\limits_{m = 1}^M {{c_{m}}} }  \le \mathop {\lim }\limits_{t \to 0} \frac{{f\left( {{P} - td} \right)\, - f\left( P \right)}}{t} \le {W_2},\]

hence there exists a constant ${c_{\min }}$ such that if ${c_{m}} \ge {c_{\min }}\,\forall m$ then the inequality above leads to a contradiction.
\end{proof}

\begin{lemma}
\label{lem:pch_equals_p_allk}
It holds for any optimal solution that
\begin{equation}
p^*_{chg,m,k} =
\begin{cases}
p^*_{m,k}, & \text{if } p^*_{m,k} \ge 0, \\
0, & \text{else}.
\end{cases}
\end{equation}
\end{lemma}

\begin{proof}
Recall the optimization problem 

\[\begin{array}{l}
\min .\;\;\;\;\;\frac{1}{2}\sum\limits_{k = 0}^{K - 1} {{{\left( {\sum\limits_{m = 1}^M {{p_{m,k}}} } \right)}^2}}  + \sum\limits_{k = 0}^{K - 1} {\sum\limits_{m = 1}^M {{p_{L,k}}{p_{m,k}}} } \\
 + \sum\limits_{k = 0}^{K - 1} {\sum\limits_{m = 1}^M {\frac{{{c_{m}}}}{a}{p_{chg,m,k}}} } \\
{\rm{s}}{\rm{.t}}{\rm{.}}\;\;\;\\
\, - \sum\limits_{m = 1}^M {{p_{m,k}}}  \le {p_{L,k}} - {p_{g,\min }},\;{\rm{for}}\;k = 0 \ldots \left( {K - 1} \right)\;\\
{\rm{for}}\;\;\;m = 1 \ldots M:\;\;\;\\
\,{p_{m,k}} - {p_{chg,m,k}} \le 0\;,\;{\rm{for}}\;\;k = 0 \ldots \left( {K - 1} \right)\;\\
\, - {p_{chg,m,k}} \le 0\;\;\;\;\;\;\;\;,\;{\rm{for}}\;\;k = 0 \ldots \left( {K - 1} \right)\;\\
\;\,\sum\limits_{j = 0}^{k - 1} {{p_{m,j}}}  - \frac{{{e_{m,\max }}}}{\Delta } \le 0,\;\;{\rm{for}}\;k = 1 \ldots K\;\\
 - \left( {\sum\limits_{j = 0}^{k - 1} {{p_{m,j}}} } \right) \le \;0,\;\;f{\rm{or}}\;k = 1 \ldots K\;\\
\,\sum\limits_{j = 0}^{K - 1} {{p_{m,j}}}  = 0\;\;\;
\end{array}\]

All the constraints are affine and Slater condition holds, hence a solution is optimal if and only if KKT conditions hold.

A solution ${x^*}$ is called a KKT point if it is feasible and there exist Lagrange multipliers $\mu  \ge 0,\,\lambda$, such that:

\begin{enumerate}
    \item ${\nabla _x}L(x,\lambda ,\mu ) = 0$ where
    
    \[\begin{array}{l}
    L(x,\lambda ,\mu ) = \\
    \frac{1}{2}\sum\limits_{k = 0}^{K - 1} {{{\left( {\sum\limits_{m = 1}^M {{p_{m,k}}} } \right)}^2}}  + \sum\limits_{k = 0}^{K - 1} {\sum\limits_{m = 1}^M {{p_{L,k}}{p_{m,k}}} } \\
     + \sum\limits_{k = 0}^{K - 1} {\sum\limits_{m = 1}^M {\frac{{{c_{m}}}}{a}{p_{chg,m,k}}} } \\
     + \sum\limits_{m = 1}^M {{\lambda _m}\sum\limits_{k = 0}^{K - 1} {{p_{m,k}}} }  + \sum\limits_{m = 1}^M {\sum\limits_{k = 0}^{K - 1} {{\mu _{m,k}}^{(1)}\left( { - {p_{chg,m,k}}} \right)} } \\
     + \sum\limits_{m = 1}^M {\sum\limits_{k = 0}^{K - 1} {{\mu _{m,k}}^{(2)}\left( {{p_{m,k}} - {p_{chg,m,k}}} \right)} } \\
     + \sum\limits_{k = 0}^{K - 1} {{\mu _k}^{(3)}\left( { - \sum\limits_{m = 1}^M {{p_{m,k}}}  + {p_{g,\min }} - {p_{L,k}}} \right)} \\
     + \sum\limits_{m = 1}^M {\sum\limits_{k = 1}^K {{\mu _{m,k - 1}}^{(4)}\left( {\sum\limits_{j = 0}^{k - 1} {{p_{m,j}}}  - \frac{{{e_{\max }}}}{\Delta }} \right)} } \\
     + \sum\limits_{m = 1}^M {\sum\limits_{k = 1}^K {{\mu _{m,k - 1}}^{(5)}\left( { - \sum\limits_{j = 0}^{k - 1} {{p_{m,j}}} } \right)} } ,
    \end{array}\]

    \item complementary-slackness conditions hold:
     
    \[\begin{array}{l}
    {\rm{for}}\,m = 1 \ldots M:\\
    \,{\mu ^*}{_{m,k}^{(1)}}\left( { - {p_{chg,m,k}}} \right) = 0,\,\,\,{\rm{for}}\,k = 0 \ldots \left( {K - 1} \right),\,\\
    {\mu ^*}{_{m,k}^{(2)}}\left( {{p_{m,k}} - {p_{chg,m,k}}} \right) = 0,\,\,\,{\rm{for}}\,k = 0\left( {K - 1} \right),\,\\
    \,{\mu ^*}{_k^{(3)}}\left( { - \sum\limits_{m = 1}^M {{p_{m,k}}}  + {p_{g,\min }} - {p_{L,k}}} \right) = 0,\,\,\,\\
    {\rm{for}}\,k = 0 \ldots \left( {K - 1} \right),\\
    {\mu ^*}{_{m,k - 1}^{(4)}}\left( {\sum\limits_{j = 0}^{k - 1} {{p_{m,j}}}  - \frac{{{e_{\max }}}}{\Delta }} \right) = 0,\,\,\,{\rm{for}}\,k = 1 \ldots K,\,\,\\
    {\mu ^*}{_{m,k - 1}^{(5)}}\left( { - \sum\limits_{j = 0}^{k - 1} {{p_{m,j}}} } \right) = 0,\,\,\,{\rm{for}}\,k = 1 \ldots K,.
    \end{array}\]

\end{enumerate}

From (1) we get that

\[\begin{array}{l}
{\rm{for}}\,\,\,k \in \{ 0,...,K - 1\} \,\,\,{\rm{and}}\,\,\,m \in \{ 0,...,M\} \\
{\left( {{\nabla _p}L(x,\lambda ,\mu )} \right)_{m,k}} = \sum\limits_{m = 1}^M {{p_{m,k}}}  + {p_{L,k}} + {\lambda _m}\\
 + {\mu _{m,k}}^{(2)} - {\mu _{m,k}}^{(3)} + \sum\limits_{j = k + 1}^K {{\mu _{m,j - 1}}^{(4)}}  - \sum\limits_{j = k + 1}^K {{\mu _{m,j - 1}}^{(5)}}  = 0,\,\,\\
{\left( {{\nabla _{{p_{chg}}}}L(x,\lambda ,\mu )} \right)_{m,k}} = \frac{{{c_{m}}}}{a} - {u_{m,k}}^{(1)} - {u_{m,k}}^{(2)} = 0,
\end{array}\]

which can be rewritten as

\[\begin{array}{l}
{\rm{for}}\,\,\,k \in \{ 0,...,K - 1\} \,\,\,{\rm{and}}\,\,\,m \in \{ 1,...,M\} :\\
\,\sum\limits_{m = 1}^M {{p_{m,k}}}  =  - {p_{L,k}} - {\lambda _m} - {\mu _{m,k}}^{(2)}\\
 + {\mu _{m,k}}^{(3)} + \sum\limits_{j = k}^{K - 1} {\left( {{\mu _{m,j}}^{(5)} - {\mu _{m,j}}^{(4)}} \right)} ,\,\\
\frac{{{c_{m}}}}{a} = {u_{m,k}}^{(1)} + {u_{m,k}}^{(2)}.
\end{array}\]

We now consider possible values of the Lagrange multipliers ${\mu _{m,k}}^{(1)}$ and ${\mu _{m,k}}^{(2)}$ and use the equations obtained from (1) and (2) to determine ${p^*}_{chg,m,k}$. Note that ${\mu _{m,k}}^{(1)} = 0,{\mu _{m,k}}^{(2)} = 0,$ is not feasible for any $\{ m,k\} $ since $\frac{{{c_{m}}}}{a} = {\mu _{m,k}}^{(1)} + {\mu _{m,k}}^{(2)}$ and ${c_{m}} > 0\,\forall m$. From complementary-slackness it follows that if ${\mu _{m,k}}^{(1)} > 0$ then ${p^*}_{chg,m,k} = 0$, and if ${\mu _{m,k}}^{(1)} = 0$ then ${\mu _{m,k}}^{(2)} > 0$ and therefore ${p^*}_{chg,m,k} = {p^*}_{m,k}$.

Recall from the problem's constraints that

\[\begin{array}{l}
{\rm{for }}k = 0 \ldots \left( {K - 1} \right),\;m = 1 \ldots M:\\
\,{p_{m,k}} - {p_{chg,m,k}} \le 0\;,\;\;\\
\, - {p_{chg,m,k}} \le 0,\;\;\;\;\;\;\;\;
\end{array}\]

and therefore if ${p^*}_{m,k} \ge 0$ then ${p^*}_{chg,\,m,k} \ge 0$, hence ${p^*}_{chg,\,m,k} = {p^*}_{m,k}$, and if ${p^*}_{m,k} < 0$ then ${p^*}_{chg,\,m,k} = 0$.

Overall, we get

\[{p^*}_{chg,m,k} = \left\{ {\begin{array}{*{20}{c}}
{{p^*}_{m,k},\,\,\,\,\,\,{\rm{if}}\,{p^*}_{m,k} \ge 0\,}\\
{0,\,\,\,\,\,\,\,\,\,\,\,\,\,\,\,\,\,\,\,\,\,\,\,\,\,\,else}
\end{array}} \right\}.\]
\end{proof}

\begin{lemma}
\label{lem:pch_equals_p_kabsorb}
Consider an optimal solution $x^*$. There exists a constant $c_{\min}$ such that if $M=2$ and $c_{\min} \le c_{m}$, $\forall m$ then
\begin{equation}
\begin{aligned}
p^*_{m,k} & =
\begin{cases}
0, & \forall k = 0, \ldots, k_1, \\
\lambda_{m,k}\left(p_{g\min} - p_{L,k}\right), & k \in K_{absorb}, \\
\le 0, & \forall k = k_2, \ldots, K - 1,
\end{cases}
\end{aligned}
\end{equation}
where $\sum_{m = 1}^M \lambda_{m,k}  = 1$, $\lambda_{m,k} \ge 0$, $\forall k \in K_{absorb}$.
\end{lemma}

\begin{proof}
Assume $M=2$ and ${c_{\min }} \le {c_{m}}\,for\,m = 1,2$.
Since 

\[0 \le {p_{chg,m,k}}^*,\,\,m = 1,2\,,k = 0...K - 1,\]

\[\sum\limits_{m = 1}^M {{p_{chg,m,k}}^*}  = {p_{g,\min }} - {p_{L,k}},\,\forall k \in {K_{absorb}}\]

(Lemma \ref{lem:optimal_charge_over_m}) and $M=2$, it follows that 
\[{p_{chg,m,k}}^* \le {p_{g,\min }} - {p_{L,k}},\,\forall k \in {K_{absorb}}.\]

In addition, based on 

\[\begin{array}{l}
\;0 < \left( {{p_{g,\min }} - {p_{L,k}}} \right) \le \sum\limits_{m = 1}^M {{p_{m,k}}}  \le \sum\limits_{m = 1}^M {{p_{chg,m,k}}} ,\,\\
\;\forall k \in {K_{absorb}}
\end{array}\] (Lemma \ref{lem:minimal_charge}) and 

\[\sum\limits_{m = 1}^M {{p_{chg,m,k}}^*}  = {p_{g,\min }} - {p_{L,k}},\,\forall k \in {K_{absorb}}\]

(Lemma \ref{lem:optimal_charge_over_m}), it follows that 

\[\sum\limits_{m = 1}^M {{p_{m,k}}^*}  = {p_{g,\min }} - {p_{L,k}},\,\forall k \in {K_{absorb}},\]

and since $M=2$ and

\[\begin{array}{l}
{p_{m,k}}^* \le {p_{chg,m,k}}^* \le {p_{g,\min }} - {p_{L,k}},\\
m = 1,2\,,\,\forall k \in {K_{absorb}}
\end{array}\]

then

\[{p_{m,k}}^* \ge 0,\,\,m = 1,2\,,k \in {K_{absorb}}.\]

Moreover, 

\[{p^*}_{chg,m,k} = \left\{ {\begin{array}{*{20}{c}}
{{p^*}_{m,k},\,\,\,\,\,\,{\rm{if}}\,{p^*}_{m,k} \ge 0\,}\\
{0,\,\,\,\,\,\,\,\,\,\,\,\,\,\,\,\,\,\,\,\,\,\,\,\,\,\,else}
\end{array}} \right\}\]

(Lemma \ref{lem:pch_equals_p_allk}) and therefore

\[{p^*}_{chg,m,k} = {p^*}_{m,k}\,\forall k \in {K_{absorb}}.\]

In addition, since ${c_{min}} \le {c_{m}}$ then it holds for the optimal solution that

\[\sum\limits_{m = 1}^M {{p_{chg,m,k}}^*}  = \max (0,{p_{g,\min }} - {p_{L,k}}),\,k = 0...K - 1\]

(Lemma \ref{lem:optimal_charge_over_m}), and therefore for $k \in {K_{absorb}}, m \in \{1,2\},$ we get 

\[{p^*}_{chg,m,k} = {\lambda _m}\left( {{p_{g\min }} - {p_{L,k}}} \right),\]

where $\sum\limits_{m = 1}^M {{\lambda _m}}  = 1$,

and since for $k \notin {K_{absorb}},m = 1,2,$ $\sum\limits_{m = 1}^M {{p^*}_{chg,m,k}}  = 0$ and ${p^*}_{chg,m,k} \ge 0$ then

\[{p^*}_{chg,m,k} = 0,\,\;{\rm{for}}\;\;k \notin {K_{absorb}}.\]

Recall that ${k_1}$ and ${k_2}$ are the minimal and maximal elements in $K_{absorb}$ accordingly. Since

\[{p^*}_{chg,m,k} = 0,\,\;{\rm{for}}\;\;k \notin {K_{absorb}}\] 

then 

\[{p^*}_{m,k} \le 0\,\forall k \notin {K_{absorb}}\]

(Lemma \ref{lem:pch_equals_p_allk}), and in particular

\[{p^*}_{m,k} \le 0,\,\forall k = {k_2} + \Delta , \ldots K.\]

Moreover, together with the constraint

\[0 \le \sum\limits_{j = 0}^{k - 1} {{p_{m,j}}} ,\;{\rm{for}}\;k = 1 \ldots K,\,m = 1 \ldots M\]

it follows that

\[{p^*}_{m,k} = 0,\,\forall k = 0 \ldots {k_1} - \Delta .\]

Overall, we get

\[\begin{array}{l}
{p^*}_{chg,m,k} = {\lambda _m}\left( {{p_{g\min }} - {p_{L,k}}} \right),\,k \in {K_{absorb}},\\
{p^*}_{chg,m,k} = 0,\,k \notin {K_{absorb}},\\
{p^*}_{m,k} = {\lambda _m}\left( {{p_{g\min }} - {p_{L,k}}} \right),k \in {K_{absorb}},\\
{p^*}_{m,k} = 0,\,\forall k = 0 \ldots {k_1} - \Delta ,\\
{p^*}_{m,k} \le 0\,\forall k = {k_2} + \Delta , \ldots K - 1.
\end{array}\]

where $\sum\limits_{m = 1}^M {{\lambda _{m,k}}}  = 1,\,{\lambda _{m,k}} \ge 0,\forall k \in {K_{absorb}}.$
\end{proof}

\begin{lemma}
\label{lem:optimal_charge_over_k}
Consider an optimal solution $x^*$. There exists a constant $c_{\min}$ such that if $M=2$ and $c_{\min} \le c_{m}$, $\forall m$ then it holds for the optimal solution that
\begin{equation}
\sum_{k = 0}^{K - 1} {{p^*}_{chg,m,k}}  = \begin{cases}
E_{m,split}, & c_{m} = c_{M/m}, \\
E_{m,\max}, & c_{m} < c_{M/m}, \\
E_{m,\min }, & c_{m} > c_{M/m},
\end{cases}
\end{equation}
where
\begin{equation}
\begin{aligned}
&E_{m,\max} = \min\left(E_{absorb},\frac{e_{\max,m}}{\Delta}\right), \\
&E_{m,\min} = \max\left(E_{absorb} - \frac{e_{\max ,M/m}}{\Delta},0\right), \\
&E_{m,split} = \\
&\min\left(\max\left(\frac{1}{2}E_{absorb},E_{absorb} - \frac{e_{\max,M/m}}{\Delta}\right),\frac{e_{\max ,m}}{\Delta}\right).
\end{aligned}
\end{equation}
\end{lemma}

\begin{proof}
Recall the original optimization problem 

\[\begin{array}{l}
\min .\;\;\;\;\;f\left( P \right)\, + h({P_{chg}})\\
s.t.\\
 - \sum\limits_{m = 1}^M {{p_{m,k}}}  \le {p_{L,k}} - {p_{g,\min }},\;{\rm{for}}\;k = 0 \ldots \left( {K - 1} \right),\;\\
{\rm{for}}\;\;\;m = 1 \ldots M:\;\\
0 \le {p_{chg,m,k}}\;,\;{\rm{for}}\;\;k = 0 \ldots \left( {K - 1} \right)\;\;\\
\,{p_{m,k}} - {p_{chg,m,k}} \le 0\;,\;{\rm{for}}\;\;k = 0 \ldots \left( {K - 1} \right)\;\\
\;\,\sum\limits_{j = 0}^{k - 1} {{p_{m,j}}}  - \frac{{{e_{m,\max }}}}{\Delta } \le 0,\;\;{\rm{for}}\;k = 1 \ldots K\;\\
 - \left( {\sum\limits_{j = 0}^{k - 1} {{p_{m,j}}} } \right) \le \;0,\;\;f{\rm{or}}\;k = 1 \ldots K\;\\
\,\sum\limits_{j = 0}^{K - 1} {{p_{m,j}}}  = 0\;\\
{\rm{where}}\\
{\left( P \right)_{MK \times 1}} = {\left( \begin{array}{l}
{p_{1,0}}, \ldots ,{p_{M,0}},{p_{1,1}}, \ldots ,{p_{M,1}},\\
 \ldots ,{p_{1,\left( {K - 1} \right)}}, \ldots ,{p_{M,\left( {K - 1} \right)}}
\end{array} \right)^{\rm{T}}}\\
{\left( {{P_{chg}}} \right)_{MK \times 1}} = {\left( {\begin{array}{*{20}{c}}
{{p_{chg,}}_{1,0}, \ldots ,{p_{chg,M,0}},}\\
{{p_{chg,1,1}}, \ldots ,{p_{chg,M,1}}, \ldots ,}\\
{{p_{chg,1,\left( {K - 1} \right)}}, \ldots ,{p_{chg,M,\left( {K - 1} \right)}}}
\end{array}} \right)^{\rm{T}}}\\
f(P) = \frac{1}{2}\sum\limits_{k = 0}^{K - 1} {{{\left( {\sum\limits_{m = 1}^M {{p_{m,k}}} } \right)}^2}}  + \sum\limits_{k = 0}^{K - 1} {\left( {{p_{L,k}}\sum\limits_{m = 1}^M {{p_{m,k}}} } \right)} \\
h({P_{chg}}) = \sum\limits_{k = 0}^{K - 1} {\sum\limits_{m = 1}^M {{c_{m}}{p_{chg,m,k}}} } 
\end{array}\]

Consider an optimal solution ${x^*}$. Based on Lemma \ref{lem:pch_equals_p_kabsorb}, there exists a constant $c_{min}$ such that if $M = 2$ and ${c_{\min }} \le {c_{m}}\,for\,m = 1,2$ then

\[\begin{array}{l}
{p^*}_{chg,m,k} = {p^*}_{m,k} = {\lambda _{m,k}}\left( {{p_{g\min }} - {p_{L,k}}} \right),\,k \in {K_{absorb}}\\
{p^*}_{chg,m,k} = 0,\,k \notin {K_{absorb}}.
\end{array}\]

where $\sum\limits_{m = 1}^M {{\lambda _{m,k}}}  = 1\,\,{\rm{and}}\,\,0 \le {\lambda _{m,k}},\,\forall k \in {K_{absorb}}.$

We denote \[\lambda  = \left\{ {{\lambda _{m,k}}|\sum\limits_{m = 1}^M {{\lambda _{m,k}}}  = 1,\,0 \le {\lambda _{m,k}},\forall k \in {K_{absorb}}} \right\}.\]

Since the expressions $\sum\limits_{m = 1}^M {{p^*}_{m,k}} ,\,k \in {K_{absorb}}$ are constant then $f(P)$ does not depend on $\lambda $. In addition, the cost function is $f(P) + h({P_{chg}})$, i.e., an additively separable function, hence the optimal solution is any feasible $\lambda $ that minimizes $h({P_{chg}})$. We now turn our attention to the sum $\sum\limits_{k = 0}^{K - 1} {{p_{chg,m,k}}} $ and note that based on Lemma \ref{lem:pch_equals_p_kabsorb},

\[\begin{array}{l}
{p^*}_{chg,m,k} = {\lambda _{m,k}}\left( {{p_{g\min }} - {p_{L,k}}} \right),\,k \in {K_{absorb}},\\
{p^*}_{chg,m,k} = 0,\,k \notin {K_{absorb}},\\
{p^*}_{m,k} = {\lambda _{m,k}}\left( {{p_{g\min }} - {p_{L,k}}} \right),k \in {K_{absorb}},\\
{p^*}_{m,k} = 0,\,\forall k = 0 \ldots {k_1} - \Delta .
\end{array}\]

and therefore for $k \in {K_{absorb}},\,m = 1,2$ the constraints

\[\begin{array}{l}
 - \sum\limits_{m = 1}^M {{p_{m,k}}}  \le {p_{L,k}} - {p_{g,\min }},\\
0 \le {p_{chg,m,k}},\\
\,{p_{m,k}} - {p_{chg,m,k}} \le 0,\\
0 \le \sum\limits_{k = 0}^{k - 1} {{p_{m,k}}} .\;
\end{array}\]

hold for any $\lambda$.

Hence, the only constraints on $\lambda $ are 

\[\sum\limits_{k = 0}^{k - 1} {{p_{m,k}}}  \le \frac{{{e_{m,\max }}}}{\Delta },\,k = 0 \ldots \left( {K - 1} \right).\]

Based on 

\[\begin{array}{l}
{p^*}_{m,k} = 0,\, k = 0 \ldots {k_1} - \Delta ,\\
{p^*}_{m,k} = {\lambda _{m,k}}\left( {{p_{g\min }} - {p_{L,k}}} \right),k \in {K_{absorb}},\\
{p^*}_{m,k} \le 0\, k = {k_2} + \Delta , \ldots K - 1,
\end{array}\]

The constraints above can be written as

\[\sum\limits_{j = {k_1}}^{k - 1} {{\lambda _{m,j}}\left( {{p_{g\min }} - {p_{L,j}}} \right)}  \le \frac{{{e_{m,\max }}}}{\Delta },k \in {K_{absorb}},\]

And since ${\lambda _{m,k}} \ge 0$ then the constraints above can be simplified to

$\sum\limits_{k \in {K_{absorb}}} {{\lambda _{m,k}}\left( {{p_{g\min }} - {p_{L,k}}} \right)}  \le \frac{{{e_{m,\max }}}}{\Delta },\,\forall m.$

Next, we will find the values of $\sum\limits_{k = 0}^{K - 1} {{p_{chg,m,k}}} $ such that ${\lambda ^*} = \arg \min (h)$ and

$\sum\limits_{k \in {K_{absorb}}} {{\lambda _{m,k}}\left( {{p_{g\min }} - {p_{L,k}}} \right)}  \le \frac{{{e_{m,\max }}}}{\Delta },m \in \{1,2\}$

for different relations between ${c_{m}}\,{\rm{and}}\,{c_{M/m}}$. Note that ${\lambda ^*}_{m,k},\,k \in {K_{absorb}},m = 1,2$ strictly define ${\lambda ^*}$. Also note that for any optimal solution the function $h$ is

\[\begin{array}{l}
h({P^*}_{chg}) = \sum\limits_{k = 0}^{K - 1} {\sum\limits_{m = 1}^M {{c_{m}}{p^*}_{chg,m,k}} }  = \\
\sum\limits_{k \in {K_{absorb}}} {({p_{g\min }} - {p_{L,k}})({\lambda _{m,k}}{c_{m}} + {\lambda _{M/m,k}}{c_{M/m}})}
\end{array}\]

(Lemma~\ref{lem:pch_equals_p_kabsorb}). Hence if ${c_{m}} < {c_{M/m}}$ where $m \in \{ 1,2\} $ then ${\lambda _{m,k}},k \in {K_{absorb}}$ should be maximized to minimize $h$, and if ${c_{m}} = {c_{M/m}}$ then $h$ does not depend on the choice of ${\lambda _{m,k}},k \in {K_{absorb}}$. 

We now consider the case that ${c_{m}} < {c_{M/m}}$, in which we want to maximize ${\lambda _{m,k}},k \in {K_{absorb}}$.

If ${E_{absorb}} < \frac{{{e_{m,\max }}}}{\Delta }$ then ${\lambda _{m,k}}^* = 1,\,k \in {K_{absorb}},m = 1,2$ and it holds for the optimal solution that

\[\begin{array}{l}
\sum\limits_{k \in {K_{absorb}}} {{\lambda ^*}_{m,k}({p_{g\min }} - {p_{L,k}})}  = \sum\limits_{k = 0}^{K - 1} {{p^*}_{chg,m,k}}  = {E_{absorb}},\\
\sum\limits_{k \in {K_{absorb}}} {{\lambda ^*}_{M/m,k}({p_{g\min }} - {p_{L,k}})}  = \sum\limits_{k = 0}^{K - 1} {{p^*}_{chg,M/m,k}}  = 0.
\end{array}\]

If $\frac{{{e_{m,\max }}}}{\Delta } < {E_{absorb}}$ then it holds for the optimal solution that

\[\begin{array}{l}
\sum\limits_{k \in {K_{absorb}}} {{\lambda _{m,k}}({p_{g\min }} - {p_{L,k}})}  = \\
\sum\limits_{k = 0}^{K - 1} {{p^*}_{chg,m,k}}  = \frac{{{e_{m,\max }}}}{\Delta },\\
\sum\limits_{k \in {K_{absorb}}} {{\lambda ^*}_{M/m,k}({p_{g\min }} - {p_{L,k}})}  = \\
\sum\limits_{k = 0}^{K - 1} {{p^*}_{chg,M/m,k}}  = {E_{absorb}} - \frac{{{e_{m,\max }}}}{\Delta },
\end{array}\]

and together we get

\[\begin{array}{l}
\sum\limits_{k = 0}^{K - 1} {{p^*}_{chg,m,k}}  = \min ({E_{absorb}},\frac{{{e_{m,\max }}}}{\Delta }),\\
\sum\limits_{k = 0}^{K - 1} {{p^*}_{chg,M/m,k}}  = \max ({E_{absorb}} - \frac{{{e_{m,\max }}}}{\Delta },0).
\end{array}\]

We now consider the case that ${c_{m}} = {c_{M/m}}$ in which we can choose any feasible ${\lambda _{m,k}},k \in {K_{absorb}}$. We may assume a 'fair' solution, i.e., ${\lambda _{m,k}},k \in {K_{absorb}}$ are as close to $\frac{1}{2}$ as the constraints allow.  

If $\frac{1}{2}{E_{absorb}} \le \frac{{{e_{\max ,m}}}}{\Delta }$ and $\frac{1}{2}{E_{absorb}} \le \frac{{{e_{\max ,M/m}}}}{\Delta }$ then ${\lambda _{m,k}}^* = 1,\,k \in {K_{absorb}},m = 1,2$ and it holds for the optimal solution that

\[\begin{array}{l}
\sum\limits_{k \in {K_{absorb}}} {{\lambda ^*}_{m,k}({p_{g\min }} - {p_{L,k}})}  = \\
\sum\limits_{k = 0}^{K - 1} {{p^*}_{chg,m,k}}  = \frac{1}{2}{E_{absorb}},\\
\sum\limits_{k \in {K_{absorb}}} {{\lambda ^*}_{M/m,k}({p_{g\min }} - {p_{L,k}})}  = \\
\sum\limits_{k = 0}^{K - 1} {{p^*}_{chg,M/m,k}}  = \frac{1}{2}{E_{absorb}},
\end{array}\]

If $\frac{1}{2}{E_{absorb}} > \frac{{{e_{\max ,m}}}}{\Delta }$ and $\frac{1}{2}{E_{absorb}} \le \frac{{{e_{\max ,M/m}}}}{\Delta }$ then the closest solution to a fair allocation that holds the constraints is

\[\begin{array}{l}
\sum\limits_{k \in {K_{absorb}}} {{\lambda _{m,k}}({p_{g\min }} - {p_{L,k}})}  = \\
\sum\limits_{k = 0}^{K - 1} {{p^*}_{chg,m,k}}  = \frac{{{e_{m,\max }}}}{\Delta },\\
\sum\limits_{k \in {K_{absorb}}} {{\lambda ^*}_{M/m,k}({p_{g\min }} - {p_{L,k}})}  = \\
\sum\limits_{k = 0}^{K - 1} {{p^*}_{chg,M/m,k}}  = {E_{absorb}} - \frac{{{e_{m,\max }}}}{\Delta },
\end{array}\]

If $\frac{1}{2}{E_{absorb}} \le \frac{{{e_{\max ,m}}}}{\Delta }$ and $\frac{1}{2}{E_{absorb}} > \frac{{{e_{\max ,M/m}}}}{\Delta }$ then the closest solution to a a fair allocation that holds the constraints is

\[\begin{array}{l}
\sum\limits_{k \in {K_{absorb}}} {{\lambda _{m,k}}({p_{g\min }} - {p_{L,k}})}  = \\
\sum\limits_{k = 0}^{K - 1} {{p^*}_{chg,m,k}}  = {E_{absorb}} - \frac{{{e_{M/m,\max }}}}{\Delta },\\
\sum\limits_{k \in {K_{absorb}}} {{\lambda ^*}_{M/m,k}({p_{g\min }} - {p_{L,k}})}  = \\
\sum\limits_{k = 0}^{K - 1} {{p^*}_{chg,M/m,k}}  = \frac{{{e_{M/m,\max }}}}{\Delta },
\end{array}\]

Together we get

\[\begin{array}{l}
\sum\limits_{k = 0}^{K - 1} {{p^*}_{chg,m,k}}  = \\
\min \left( {\max \left( {\frac{1}{2}{E_{absorb}},{E_{absorb}} - \frac{{{e_{\max ,M/m}}}}{\Delta }} \right),\frac{{{e_{\max ,m}}}}{\Delta }} \right),\\
\sum\limits_{k = 0}^{K - 1} {{p^*}_{chg,M/m,k}}  = \\
\min \left( {\max \left( {\frac{1}{2}{E_{absorb}},{E_{absorb}} - \frac{{{e_{\max ,m}}}}{\Delta }} \right),\frac{{{e_{\max ,M/m}}}}{\Delta }} \right),
\end{array}\]

and overall we get for $m = \{ 1,2\} $ that it holds for the optimal solution that

\[\begin{array}{l}
\sum\limits_{k = 0}^{K - 1} {{p^*}_{chg,m,k}}  = \left\{ {\begin{array}{*{20}{c}}
{{E_{m,split}},\,\,\,{c_{m}} = {c_{M/m}}}\\
{{E_{m,\max }},\,\,\,{c_{m}} < {c_{M/m}}}\\
{{E_{m,\min }},\,\,\,\,{c_{m}} > {c_{M/m}}}
\end{array}} \right\}\\
{\rm{where}}\,\\
{E_{m,\max }} = \min \left( {{E_{absorb}},\frac{{{e_{\max ,m}}}}{\Delta }} \right),\\
{E_{m,\min }} = \max \left( {{E_{absorb}} - \frac{{{e_{\max ,M/m}}}}{\Delta },0} \right),\\
{E_{m,split}} = \\
\min \left( {\max \left( {\frac{1}{2}{E_{absorb}},{E_{absorb}} - \frac{{{e_{\max ,M/m}}}}{\Delta }} \right),\frac{{{e_{\max ,m}}}}{\Delta }} \right).
\end{array}\]
\end{proof}

\begin{lemma}
\label{lem:optimal_c_diff}
There exists a constant $c_{\min}$ such that if $c_{\min} \le c_{m}$, $m = \{1,2\}$ and $c_{\min} < \frac{E_{m,\min}}{E_{m,\max}} \cdot c_{\max}$, $m = \{1,2\}$ then storage $m\in \{1,2\}$'s optimal decision is
\begin{equation}
c^*_{m} = \begin{cases}
c_{M/m} - \delta, & c_{\max}E_{m,\min} \le \left(c_{M/m} - \delta\right)E_{m,\max}, \\
c_{\max}, & c_{\max}E_{m,\min}> \left(c_{M/m} - \delta\right)E_{m,\max},
\end{cases}
\end{equation}
where
\begin{equation}
\begin{aligned}
E_{m,\max} &= \min\left(E_{absorb},\frac{e_{\max,m}}{\Delta} \right),\\
E_{m,\min} &= \max\left(E_{absorb} - \frac{e_{\max,M/m}}{\Delta},0\right).
\end{aligned}
\end{equation}
\end{lemma}

\begin{proof}
By the Lemma's assumption, ${c_{\min }} \le {c_{m}} $ hence 

\[\begin{array}{l}
\sum\limits_{k = 0}^{K - 1} {{p^*}_{chg,m,k}}  = \left\{ {\begin{array}{*{20}{c}}
{{E_{m,split}},\,\,\,{c_{m}} = {c_{M/m}}}\\
{{E_{m,\max }},\,\,\,{c_{m}} < {c_{M/m}}}\\
{{E_{m,\min }},\,\,\,\,{c_{m}} > {c_{M/m}}}
\end{array}} \right\}\\
{\rm{where}}\,\\
{E_{m,\max }} = \min \left( {{E_{absorb}},\frac{{{e_{\max ,m}}}}{\Delta }} \right),\\
{E_{m,\min }} = \max \left( {{E_{absorb}} - \frac{{{e_{\max ,M/m}}}}{\Delta },0} \right),\\
{E_{m,split}} = \\
\min \left( {\max \left( {\frac{1}{2}{E_{absorb}},{E_{absorb}} - \frac{{{e_{\max ,M/m}}}}{\Delta }} \right),\frac{{{e_{\max ,m}}}}{\Delta }} \right),
\end{array}\]

(Lemma~\ref{lem:optimal_charge_over_k}), and it follows that the cost function of each storage system is 
\[\begin{array}{l}
{J_m}({c_{m}}) = {c_{m}}\sum\limits_{k = 0}^{K - 1} {{p_{chg,m,k}}}  = \\
\left\{ {\begin{array}{*{20}{c}}
{{c_{m}} \cdot {E_{m,split}},\,\,\,\,\,{c_{m}} = {c_{M/m}}}\\
{{c_{m}} \cdot {E_{m,\max }},\,\,\,\,\,\,{c_{m}} < {c_{M/m}}}\\
{{c_{m}} \cdot {E_{m,min}},\,\,\,\,\,\,\,{c_{m}} > {c_{M/m}}}
\end{array}} \right\}\,.
\end{array}\]

% We also assume that ${c_{\min }} < \frac{{{E_{\min }}}}{{{E_{\max }}}} \cdot {c_{\max }}$\footnote{From here on we drop the index $m$ in $E_{m,\min}$ and $E_{m,\max}$ referring to the player if the reference is clear.}, and since ${E_{\min }} < {E_{\max }}$ then ${c_{\min }} < \frac{{{E_{\min }}}}{{{E_{\max }}}} \cdot {c_{\max }} < {c_{\max }}$, hence ${c_{M/m}} - \delta $ is either larger, smaller or equal to $\frac{{{E_{\min }}}}{{{E_{\max }}}} \cdot {c_{\max }}$. 

From here until the end of this proof, we use the notation $E_{\min}$ for $E_{m,\min}$ and $E_{\max}$ for $E_{m,\max}$. 

Next, we show that

\[\begin{array}{l}
{J_m}({c^*}_{m}) > {J_m}({c_{m}}),\,\,\\
{c_{m}} \in \left\{ {{c_{\min }}, \ldots ,{c_{\max }}} \right\} / {c^*}_{m}
\end{array}\]

where

\[\begin{array}{l}
{c^*}_{m} = \\
\left\{ {\begin{array}{*{20}{c}}
{{c_{M/m}} - \delta ,\,\,\left( {{c_{M/m}} - \delta } \right) \ge \frac{{{E_{\min }}}}{{{E_{\max }}}} \cdot {c_{\max }}\,\,\,\,\,\,\,\,}\\
{{c_{\max ,}}\,\,\,\,\,\,\,\,\,\,\,\,\,\,\,\,\,\,\,\,\,\,\,\,\,\,\,\,\left( {{c_{M/m}} - \delta } \right) < \frac{{{E_{\min }}}}{{{E_{\max }}}} \cdot {c_{\max }}}
\end{array}} \right\}
\end{array}\].

If $\,\left( {{c_{M/m}} - \delta } \right) < \frac{{{E_{\min }}}}{{{E_{\max }}}} \cdot {c_{\max }}$ then for ${c_{m}} = {c_{\min }}, \ldots ,\left( {{c_{M/m}} - \delta } \right),$

\[{J_m}({c^*}_{m}) = {c_{\max }} \cdot {E_{\min }} > {c_{m}} \cdot {E_{\max }} = {J_m}({c_{m}}).\]

For ${c_{m}} = \left( {{c_{M/m}} + \delta } \right), \ldots ,\left( {{c_{\max }} - \delta } \right),$ 

\[{J_m}({c^*}_{m}) = {c_{\max }} \cdot {E_{\min }} > {c_{m}} \cdot {E_{\min }} = {J_m}({c_{m}}).\]

For ${c_{m}} = {c_{M/m}}$, since \[{c_{m}} \ge {c_{\min }} \gg \delta \] and ${E_{\max }} > {E_{split}}$ then 

\[\begin{array}{l}
{J_m}({c^*}_{m}) = {c_{\max }} \cdot {E_{\min }} > ({c_{M/m}} - \delta ) \cdot {E_{\max }} \cong \\
{c_{M/m}} \cdot {E_{\max }} > {c_{M/m}} \cdot {E_{split}} = {J_m}({c_{m}}).
\end{array}\]

If $\,\left( {{c_{M/m}} - \delta } \right) \ge \frac{{{E_{\min }}}}{{{E_{\max }}}} \cdot {c_{\max }}$ then for ${c_{m}} = \left( {{c_{M/m}} + \delta } \right), \ldots ,{c_{\max } -  \delta},$ 

\[\begin{array}{l}
{J_m}({c^*}_{m}) = \left( {{c_{M/m}} - \delta } \right) \cdot {E_{\max }} > \\
{c_{m}} \cdot {E_{\min }} = {J_m}({c_{m}}).
\end{array}\]

For ${c_{m}} = {c_{\min }}, \ldots ,\left( {{c_{M/m}} - 2\delta } \right),$

\[\begin{array}{l}
{J_m}({c^*}_{m}) = \left( {{c_{M/m}} - \delta } \right) \cdot {E_{\max }} > \\
{c_{m}} \cdot {E_{\max }} = {J_m}({c_{m}}).
\end{array}\]

For ${c_{m}} = {c_{M/m}}$, since \[{c_{m}} \ge {c_{\min }} \gg \delta \] and ${E_{\max }} > {E_{split}}$ then 

\[\begin{array}{l}
{J_m}({c^*}_{m}) = \left( {{c_{M/m}} - \delta } \right) \cdot {E_{\max }} \cong \\
{c_{M/m}} \cdot {E_{\max }} > {c_{M/m}} \cdot {E_{split}} = {J_m}({c_{m}}).
\end{array}\]

Finally, if $\,\left( {{c_{M/m}} - \delta } \right) \ge \frac{{{E_{\min }}}}{{{E_{\max }}}} \cdot {c_{\max }}$ and ${c_{m}} = {c_{\max }}$ then 

\[{J_m}({c^*}_{m}) = \left( {{c_{M/m}} - \delta } \right) \cdot {E_{\max }} = {c_{\max }} \cdot {E_{\min }} = {J_m}({c_{m}}),\]

and based on our assumption that for an identical profit players choose the lower price bid,

\[{c^*}_{m} = {c_{M/m}} - \delta \].
\end{proof}

We now proceed to the main result that provides analytical conditions for the divergence of the best-response equilibrium. 
\begin{theorem}
\label{theo:infinite_loop}
There exists a constant $c_{\min}$ such that if
\begin{equation}
c_{\min} < \frac{E_{m,\min}}{E_{m,\max}} \cdot c_{\max},\quad m = \{1,2\},
\end{equation}
and $c_{\min} \le c_1^{(0)} - \delta$ then the best-response algorithm enters an infinite loop.
\end{theorem}

\begin{proof} 
Recall our assumption that the players follow best-response dynamics and that the action of player $m$ in iteration $t$ of the best-response dynamics is denoted as $c_m^{(t)}$.

Recall that based on Lemma~\ref{lem:optimal_c_diff} if $c_{\min} \le c_m^{(t)}$ and $c_{\min} < L_m$, $m \in \{ 1,2\}$ then storage $m \in \{ 1,2\}$'s best response at iteration $t$ is 
\begin{equation}
c_m^{(t)} =\begin{cases}
c_{M/m}^{(t - 1)} - \delta, & L_m \le \left(c_{M/m}^{(t - 1)} - \delta\right) \\
c_{\max}, & L_m > \left(c_{M/m}^{(t - 1)} - \delta \right),
\end{cases}
\end{equation}
where
\begin{equation}
\begin{aligned}
L_m &= c_{\max}\frac{E_{m,\min}}{E_{m,\max}}, \\
E_{m,\max} &= \min\left(E_{absorb},\frac{e_{\max,m}}{\Delta}\right), \\
E_{m,\min} &= \max\left(E_{absorb} - \frac{e_{\max,M/m}}{\Delta},0\right).
\end{aligned}
\end{equation}

We assume, without loss of generality, that storage firm $m=1$ offers a price bid first, i.e., storage firm $1$ operates in even iteration indices of the best response dynamics ($t=2n$) and storage firm $2$ in odd indices ($t=2n+1$), with $n \in \mathbb{Z}_{\ge 0}$. 

Consider an arbitrary even iteration of the best response dynamics $t_0 = 2n_0$, where $n_0 \in \mathbb{Z}_{\ge 0}$, we prove by induction on $n$ that \\
\textbf{Induction proposition}: if $c_1^{(t_0)} - 2n\delta \ge L_1$ and $c_1^{(t_0)} - (2n + 1)\delta \ge L_2$ then $c_1^{(t_0 + 2n)} = c_1^{(t_0)} - 2n\delta$ and $c_2^{(t_0 + 2n + 1)} = c_1^{(t_0)} - (2n + 1)\delta$.

\textbf{Induction base (n=0)}: trivially, $c_1^{(t_0)} = c_1^{(t_0)} - 2n\delta$. Additionally, if $c_1^{(t_0)} - \delta \ge L_2$ then by Lemma~\ref{lem:optimal_c_diff} $c_2^{(t_0 + 1)} = c_1^{(t_0)} - \delta = c_1^{(t_0)} - (2n + 1)\delta$.

\textbf{Induction step $(n=k, n=k+1)$}:

\textit{Assumption for $n=k$}: assume that if $c_1^{(t_0)} - 2k\delta \ge L_1$ and $c_1^{(t_0)} - (2k + 1)\delta \ge L_2$ then $c_1^{(t_0 + 2k)} = c_1^{(t_0)} - 2k\delta$ and $c_2^{(t_0 + 2k + 1)} = c_1^{(t_0)} - (2k + 1)\delta$.

\textit{Proof for $n=k+1$}: we note that if the ``if'' condition in the induction proposition holds for $n=k+1$, i.e., $c_1^{(t_0)} - 2(k + 1)\delta \ge L_1$ and $c_1^{(t_0)} - (2(k + 1) + 1)\delta \ge L_2$ then it also holds for $n=k$, i.e., $c_1^{(t_0)} - 2k\delta \ge L_1$ and $c_1^{(t_0)} - (2k + 1)\delta \ge L_2.$ Therefore, based on the inductive assumption for $n=k$, $c_1^{(t_0 + 2k)} = c_1^{(t_0)} - 2k\delta$ and $c_2^{(t_0 + 2k + 1)} = c_1^{(t_0)} - (2k + 1)\delta$. Subtracting $\delta$ from both sides of the last equation yields $c_2^{(t_0 + 2k + 1)} - \delta  = c_1^{(t_0)} - 2(k + 1)\delta \ge L_1$, where the last inequality stems from the ``if'' condition of the induction proposition for $n=k+1$. Thus, by Lemma~\ref{lem:optimal_c_diff}, $c_1^{(t_0 + 2(k + 1))} = c_1^{(t_0)} - 2(k + 1)\delta$. Similarly, for storage firm $m=1$, subtracting $\delta$ from both sides of the last equation yields $c_1^{(t_0 + 2(k + 1))} - \delta  = c_1^{(t_0)} - (2(k + 1) + 1)\delta \ge L_2,$ where the last inequality stems from the ``if'' condition of the induction proposition for $n=k+1$. Thus, by Lemma~\ref{lem:optimal_c_diff}, $c_2^{(t_0 + 2(k + 1) + 1)} = c_1^{(t_0)} - (2(k + 1) + 1)\delta$.

Recall our assumption that storage firm $m=1$ offers the first price bid. Table~\ref{tbl:BestResponses} depicts all three possible options for relations between $c_1^{(0)}$ to $L_1$ and $c_1^{(0)} - \delta$ to $L_2$ and the best responses in iterations 1-3 based on Lemma~\ref{lem:optimal_c_diff} and the assumption that $c_{\max} - 2\delta \ge \max \{L_1,L_2\}$.

\begin{table}[htbp]
\caption{Price bids of storage firms $1$ and $2$ in iterations 1-3 of the best-response dynamics for three cases of parameter relations.}
\label{tbl:BestResponses}
\renewcommand{\arraystretch}{1.2}
\centering
\begin{tabular}{lL{1.8cm}ccc}
\toprule
Case & Parameter & ${c_2^{(1)}}$ & ${c_1^{(2)}}$ & ${c_2^{(3)}}$ \\
& relations &  \\
\midrule
1.1 & $c_1^{(0)} \ge L_1$, $c_1^{(0)} - \delta \ge L_2$ & $c_1^{(0)} - \delta$ & -- & --  \\
\midrule
1.2 & $c_1^{(0)} < L_1$, $c_1^{(0)} - \delta \ge L_2$ & $c_1^{(0)} - \delta$ & $c_{\max}$ & $c_{\max} - \delta$ \\
\midrule
1.3 & $c_1^{(0)} - \delta < L_2$ & $c_{\max}$ & $c_{\max} - \delta$ & $c_{\max} - 2\delta$ \\
\bottomrule
\end{tabular}
\end{table}

In all three cases, the ``if'' condition of the induction proposition holds true for $n=0$ and ${t_0} = {t_A}$, where either $t_A = 0$ (case 1.1) or $t_A = 2$ (cases 1.2 and 1.3). It follows from the induction proposition that the best responses in iterations $t \ge t_A$ monotonically decrease until
\begin{equation}
\label{eq:first_stop}
{c_1^{(t_A)}} - 2\left(t - t_A\right)\delta < \max\{L_1,L_2+ \delta\},
\end{equation}
i.e., for all $t_1,t_2$ in the range $\{t_A \le t \le t_A + \frac{c_1^{(t_A)} - \max\{L_1,L_2 + \delta\}}{2\delta}\}$ s.t. $t_1<t_2$ it holds that $c_m^{(t_1)} > c_m^{(t_2)}$.

% \[\begin{array}{l}
% {t_1} < {t_2} \Rightarrow {c^{({t_1})}}_m > {c^{({t_2})}}_m,m \in \{ 1,2\} ,\\
% \forall {t_1},{t_2} \in \left\{ {{t_A} \le t \le {t_A} + \frac{{{c_1^{({t_A})}} - \max \{ {L_1},{L_2} + \delta \} }}{{2\delta }}} \right\}.
% \end{array}\]

We denote the first iteration in which inequality \eqref{eq:first_stop} holds as $\tilde t$. By Lemma~\ref{lem:optimal_c_diff}, there are two options: \\
Case 2.1: if $L_1 \ge L_2 + \delta$ then $c_1^{(\tilde t)} = c_{\max}$ and $c_2^{(\tilde t + 1)} = c_{\max} - \delta$.

Case 2.2: if $L_1 < L_2 + \delta$, then $c_2^{(\tilde t + 1)} = c_{\max}$, $c_1^{(\tilde t + 2)} = c_{\max} - \delta$ and $c_2^{(\tilde t + 3)} = c_{\max} - 2\delta$.

In both cases, since $c_{\max} - 2\delta \ge \max \{L_1,L_2\} $, again it follows the ``if'' condition of the induction proposition holds true for $n=0$ and $t_0 = t_B$, where either $t_B = \tilde t$ (case 2.1) or $t_B = \tilde t + 2$ (case 2.2). It follows from the induction proposition that the best responses again monotonically decrease until 
\begin{equation}
c_1^{(t_B)} - 2\left(t - t_B\right)\delta < \max\{L_1,L_2 + \delta\}.    
\end{equation}
 
Since $\max\{L_1,L_2\} \le c_1^{(t_A)} \le c_{\max}$ then there exists $\tilde t_1$ such that $c_1^{({\tilde t}_1)} = c_1^{(t_A)}$, and the best responses continue in an infinite loop.  
\end{proof}

The game dynamics that are formally described in the proof of Theorem~\ref{theo:infinite_loop} are illustrated in Fig.~\ref{fig:loop_illustration}. To explain the game dynamics we use the following example. Assume that we have $\Delta E_{absorb}$ excess energy in the system. Also, assume we have two storage firms, and that there is more excess energy in the system than the capacity of firm $2$ but less than the capacity of firm $1$, i.e., $e_{1,\max} > \Delta E_{absorb} > e_{2,\max}$.

Figure~\ref{fig:loop_illustration} depicts the optimal storage usage for different price bids. When the price bids are low, the system operator uses both storage systems at maximal capacity (see Area A). When firm $2$ increases its price bid, its storage usage decreases (see Area B). The same is true for firm $1$, but because there is more excess energy in the system than the capacity of storage $2$, then the operator must store $\Delta E_{absorb} - e_{2,\max}$ energy in storage~$1$ (see Area C). In between low and high price bids, there is a non-linear decrease in storage usage, due to the non-linearity of conventional generation's cost (see Area D). 
% The game dynamics that lead to instability occur after the non-linear decrease in storage usage, namely, between $c_{min}$ and $c_{max}$, when storage usage is minimal (see Areas E1 and E2). 
In Areas E1 and E2, the price bids of both firms are high, so the operator stores only the excess energy in the system ($\Delta E_{absorb}$). The question is, how it will split it between the two firms? When $c_1 < c_2$, all excess energy will be stored in storage $1$ (see Area E1). When  $c_1 > c_2$, storage $2$ will be used to the max, and $\Delta E_{absorb} - e_{2,\max}$ will be stored in storage $1$ (see Area E2). On the diagonal between Areas E1 and E2, when $c_1 = c_2$, we may assume that the energy will be split half-half or perhaps in proportion to their size. 

Figure~\ref{fig:loop_illustration} also depicts the game dynamics that lead to instability. Assume that in the first iteration of the game firm $1$ offers a price bid higher than $c_{\min}$, e.g., $c_1^{(0)} = c_{\max}$. The best reply of firm $2$ will be to offer a slightly lower price to get all the excess energy, i.e., $c_2^{(1)} = c_{\max} - \delta$. In response, firm $1$ will also offer a slightly lower price, i.e., $c_1^{(2)} = c_{\max} - 2\delta$. This will continue iteratively, until at some iteration $(k)$, firm $1$ could either decrease the price again by $\delta$ to get all of the excess energy and profit $c_1^{(k)}\cdot \Delta E_{absorb}$, or, since a $\Delta E_{absorb} - e_{2,\max}$ amount of energy is guaranteed to be sold to it, it can offer a maximal price and profit $c_{\max}\cdot (\Delta E_{absorb} - e_{2,\max})$. 
% In Theorem~\ref{theo:infinite_loop} we prove that if ${c_{\min }} < \frac{{{E_{m,\min }}}}{{{E_{m,\max }}}} \cdot {c_{\max }},\,m = \left\{ {1,2} \right\}$, then firm $1$ will offer a maximal price. 
If $c_{\min} \cdot \Delta E_{absorb} <  c_{\max} \cdot (\Delta E_{absorb} - e_{2,\max})$, then firm $1$ will offer a maximal price. In return, firm $2$ will offer $c_{\max} - \delta$, which brings us to the starting point, and, as we prove for the general case in Theorem~\ref{theo:infinite_loop}, this loop will continue infinitely and the game will not converge to an equilibrium point. 

 \begin{figure*}[htbp]
     \centering
     \includegraphics[width=0.85\textwidth]{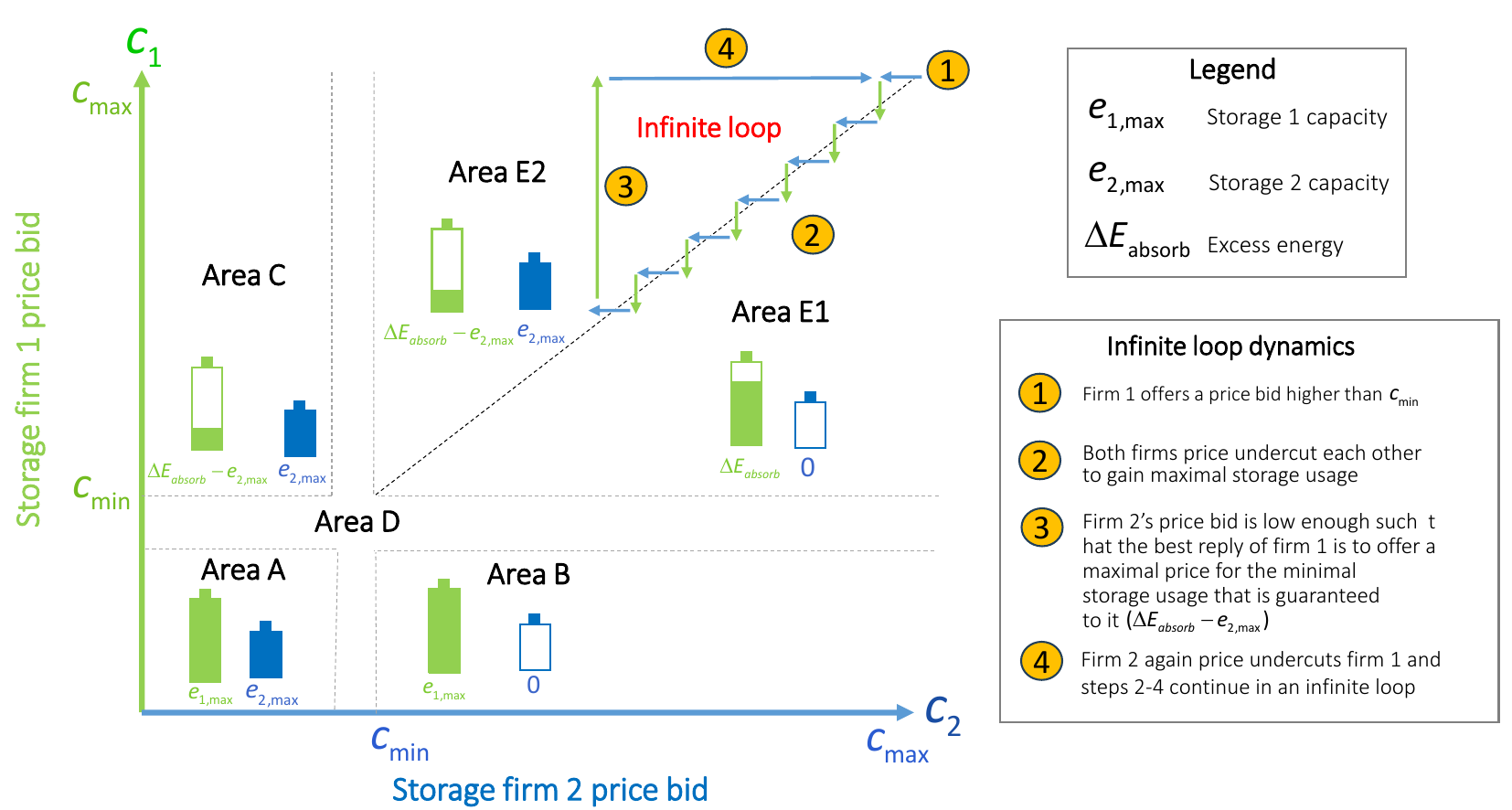}
     \caption{Illustration of the optimal storage usage per storage price bids and the game's best-reply dynamics that lead to an infinite loop.}
     \label{fig:loop_illustration}
 \end{figure*}

\section{Simulation Results}
\label{sec:sim_results}

We solve the game with a best-response algorithm (see Algorithm~\ref{alg:BR}) for a California case study. We examine two storage market structures, namely monopoly and duopoly, i.e., when one or two storage firms dominate the market, and a wide range of solar penetration, storage capacity, and system flexibility values, i.e., conventional power plants' minimal output. In the case of a Duopoly, we assume an asymmetric market, where the ratio of storage capacity of the first storage firm to the overall storage capacity, i.e., $ess_{frac}$, is $\frac{2}{3}$.

\begin{algorithm}
\caption{Best-response Algorithm}\label{alg:BR}
\begin{algorithmic}[1]
\State \textbf{Input: initial strategies $c_m^{(0)} \,  \forall m \in \{1, \ldots, M\}$, $x^{(0)}$ and simulation parameters (see Table~\ref{tbl:sim_param})} 
\State \textbf{Output: $c_m \forall m \in \{1, \ldots, M\}$ and $x$ at a NE (if exists)} 
\ForAll{$p_{solar} \in P_{solar} \bigtimes ess_{cap} \in ESS_{cap} \bigtimes p_{g,min} \in P_{g,min}$}
            \ForAll{$m \in \{1, \ldots, M\}$}
                \State $t = 0$ 
                \State $es_{max,m} = ESS_{frac}(m) \cdot ess_{cap}$
                \State Initialize strategies 
                \While{$t < $ maximum iterations} 
                    \State Set $x^{(t)}$ by solving problem~\eqref{eq:ISO_opt}.
                    \ForAll{$m \in \{1, \ldots, M\}$}
                        \State Set $c_m^{(t)}$ by solving problem~\eqref{eq:storage_opt}.
                    \EndFor
                    \If{$x^{(t)} = x^{(t-1)}$}
                    \If{$c_m^{(t)} = c_m^{(t-1)} \, \forall m \in \{1, \ldots, M\}$}
                        \State Return $x^{(t)}$ and $c_m^{(t)} = c_m^{(t-1)} \, \forall m \in \{1, \ldots, M\}$  \Comment{Converged to NE}
                        \EndIf
                    \EndIf
                    \State $t = t + 1$
                \EndWhile
            \EndFor
\EndFor
\end{algorithmic}
\end{algorithm}

The simulation was applied to California Independent System Operator (CAISO) load and generation data \cite{CAISO2020}. A typical meteorological year (TMY) was chosen, 2021, and daily load and generation profiles from March 13th, which had a combination of low demand and high solar generation. We study a range of solar penetration values that, on one hand, imply that the net load is lower than the minimal conventional generation, i.e., there is excess generated energy, and, on the other hand, enable a feasible solution, i.e., the daily sum of minimal conventional generation and solar generation is smaller than the overall daily energy demand. To this end, we vary the share of solar energy out of the overall daily energy demand, i.e., $p_{solar}$, from 30\% to 70\% with intervals of 10\%. Furthermore, we vary the share of energy storage capacity out of the excess generated solar energy, i.e., $ess_{cap}$, from 120\% to 300\% with intervals of 20\%. 
%, i.e., a 200\% share of storage capacity means that there is twice the amount of storage than the minimal amount that is required to obtain energy balance without curtailment.

We examine low and high flexibility scenarios, where each flexibility scenario is defined by the minimal output of conventional power plants. Following the flexibility analysis in \cite{Arbabzadeh2019}, we assume a minimal generation that is 12.5\% and 25\% of the peak demand in the low and high flexibility scenarios, respectively. Simulation parameters are summarized in Table~\ref{tbl:sim_param}.

\begin{table}[H]
\caption{Parameters for the numeric analysis.}
\label{tbl:sim_param}
\renewcommand{\arraystretch}{1.2}
\centering
\begin{tabular}{ll}
\toprule
Parameter & Value \\
\midrule
$p_{solar}$ & 30\%-70\% with 10\% intervals \\
$ess_{frac}$ & $\frac{2}{3}$ if $M=2$ (Duopoly) and $1$ if $M=1$ (Monopoly) \\
$ess_{cap}$ & 120\%-300\% with 20\% intervals  \\
$\frac{p_{g,min}}{max(p_{L,k})}$ & 0.125, 0.25 \\
\bottomrule
\end{tabular}
\end{table}

% We assume a minimal generation that is 25\% and 50\% of the peak demand in the low and high flexibility scenarios, respectively. 
% The quadratic cost function coefficient of conventional generation is $a = 1$ (see Eq. \ref{eq:quad_cost_f}). Recall that the cost function coefficient $b$ was omitted during the problem reformulation since it does not affect the optimal solution. 

% As can be seen in Fig.~\ref{fig:Instability}, in all combinations of solar and storage in which Theorem~\ref{theo:infinite_loop} predicts an infinite loop, indeed the best-response algorithm does not converge to an equilibrium. Moreover, there are additional combinations of solar and storage that Theorem~\ref{theo:infinite_loop} does not predict, but the simulation predicts instability nonetheless (recall that Theorem~\ref{theo:infinite_loop} is an ``if'' statement rather than ``if and only if''). 

% \subsection{Price of Storage Usage}

% We assume that the minimal price bid that a storage firm can offer is the cost of storage operation and the maximal bid is the cost of alternative flexibility solutions, e.g., curtailment and demand response, or a market price cap, the lower between them. 

% We study the existence of a stable solution, i.e., Nash Equilibrium, in a game between the system operator and energy storage firms. We analyze two scenarios, one and two storage firms, and a range of values of the minimal conventional generation in the system (``must-run'' limitation). 

Figure~\ref{fig:Price_outcome} provides an overview of our game model simulation results. It describes the cost of storage usage for both market structures and for combinations of low and high system flexibility and storage capacity values.

\begin{figure}[htbp]
\centering
\includegraphics[width=0.8\columnwidth]{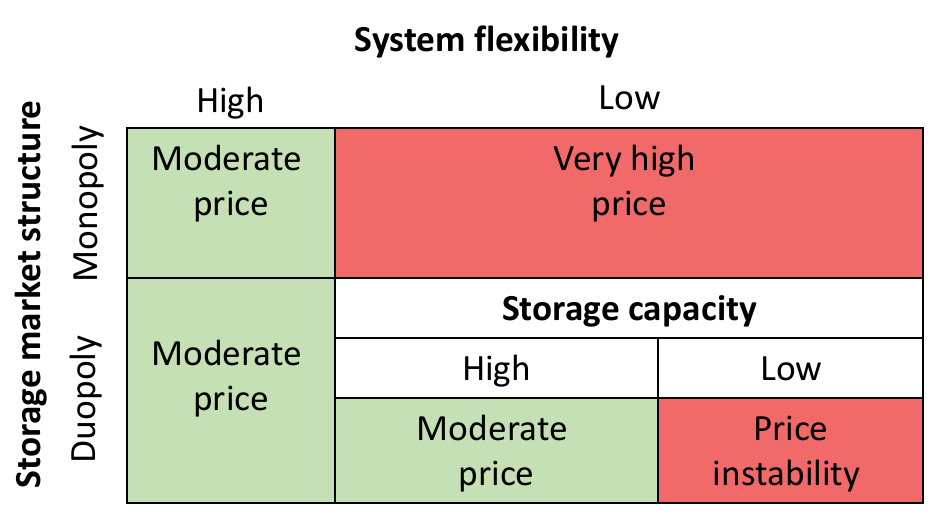}
\caption{Storage price bids at equilibrium for monopoly and duopoly market structures and for combinations of system flexibility and storage capacity values. Price instability means that the best reply dynamics do not converge to an equilibrium.}
\label{fig:Price_outcome}
\end{figure}

For both examined market structures, as long as the system flexibility is high enough, the price bids at equilibrium are moderate, i.e., they are equal to the marginal cost of conventional electricity generation, due to the competition between storage firms and power plants. Nonetheless, if system flexibility is low and there is excess generated electricity in the system, then in a monopoly, the storage firm requests a maximal price bid to store this excess energy. This maximal storage price bid is capped by either a price cap regulation or the cost of storage alternatives, e.g., curtailment or demand response. In a duopoly, if both the system flexibility and storage capacity are low, the price might be unstable, i.e., the game might not converge to an equilibrium (see Theorem~\ref{theo:infinite_loop} for the exact conditions).

Figure~\ref{fig:Instability} illustrates the existence of a stable solution (NE), corresponding to various combinations of solar and storage capacity values, for both high and low flexibility values. As we increase solar generation the system becomes less stable. However, an increase in storage capacity and system flexibility contributes to an improvement in system stability. These findings align with the analytical conditions for instability outlined in Theorem~\ref{theo:infinite_loop}.

We also find that the stability of the game can be improved without investing in system infrastructure, simply by changing the rules of the game. Figure~\ref{fig:instability_price_cap} presents the percentage of stable solutions, out of the same ensemble of scenarios that were examined in the experiment described above and presented in Fig.~\ref{fig:Instability}, but for 6 different price caps. As evident in Fig.~\ref{fig:instability_price_cap}, the share of stable solutions increases as the price cap decreases. 
% The price cap is presented in arbitrary units, where a price cap of $1$ is equivalent to the highest storage price bid that ensures maximal storage usage. 
% By lowering the price cap, we achieve a higher percentage of stable solutions. 

\begin{figure}[htbp]
\centering
\includegraphics[width=\columnwidth]{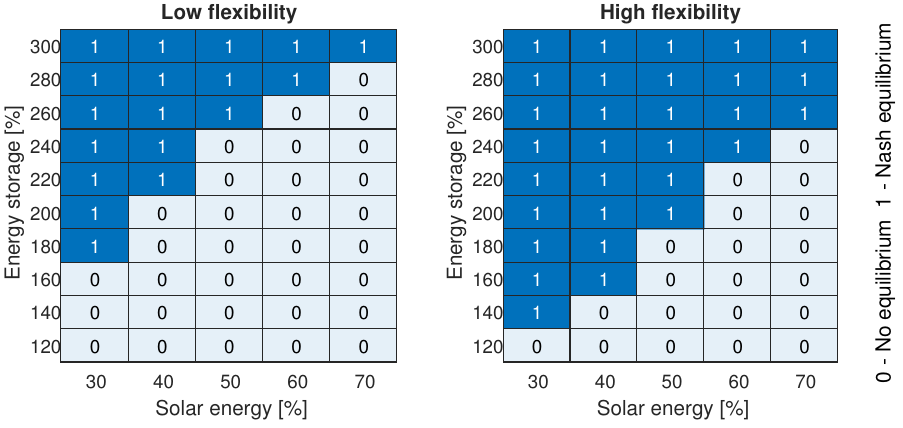}
\caption{Existence of a stable solution as a function of the share of solar energy and energy storage for low and high flexibility scenarios. The share of solar energy is expressed relative to the overall daily energy demand, while the share of storage capacity is expressed relative to the daily excess generated solar energy.}
\label{fig:Instability}
\end{figure}

\begin{figure}[htbp]
\centering
\includegraphics[width=0.8\columnwidth]{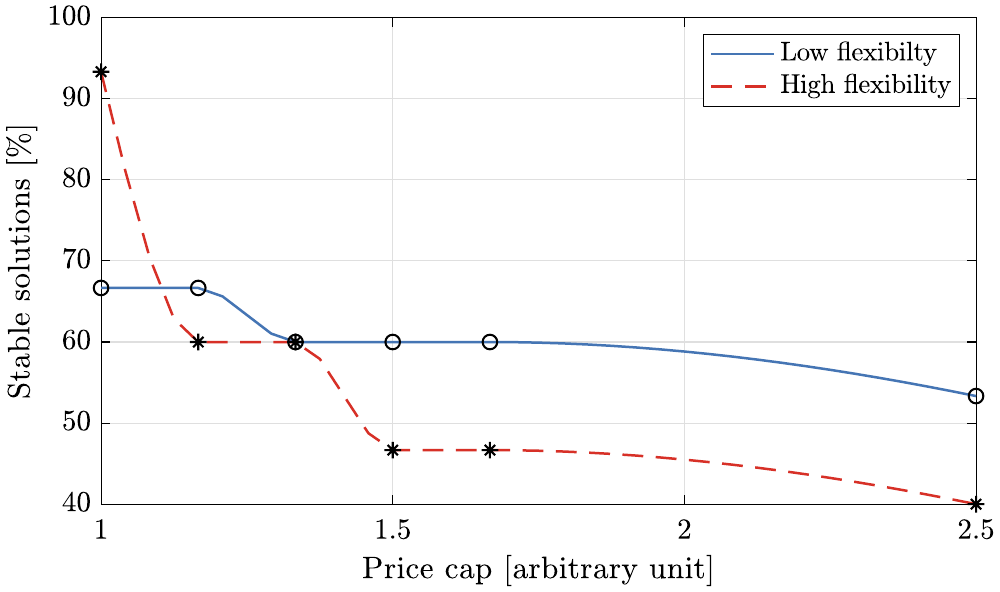}
\caption{The share of stable solutions as a function of the price cap on storage price bids for low and high flexibility scenarios.}
\label{fig:instability_price_cap}
\end{figure}

% We suggest three solutions to avoid instability: (a) limit the possible storage price bids, (b) secure a cost for the minimal energy that is required through availability payment (the storage firms offer a monthly/yearly cost, rather than a day-ahead price bidding), and (c) allow the grid operator to choose the storage firm at random if the prices that they offer are too high. With solution (c), the average utility of the storage firms for high price bids will decrease, and the loop will be avoided. 

% \subsection{Profitability of Energy Arbitrage} 

Finally, we analyze the profitability of utilizing energy storage for two different applications: energy balancing and energy arbitrage. To this end, three solar penetration scenarios (0\%, 5\%, and 10\% of daily demand) are considered, with the overall storage capacity set at 150\% of excess generated solar energy. Subplots (a)-(c) in Fig.~\ref{fig:storage_usage} depict heat maps showing the share of storage used for energy arbitrage relative to the capacity not used for balancing, against storage price bids. As solar penetration increases, the maximal price bid that the grid operator accepts for energy arbitrage decreases. In other words, the profitability of energy arbitrage diminishes with increasing solar penetration.

An explanation for this trend can be found in subplots (e)-(f) in Fig.~\ref{fig:storage_usage}, which show for each solar penetration scenario the power profiles of conventional generation plus storage in two cases. In case 1, the storage price bid is maximal hence storage usage is minimal (used for balancing purposes only), whereas in case 2, the storage price bid is minimal hence storage usage is maximal (used for both balancing and arbitrage). 

As can be seen in Fig.~\ref{fig:storage_usage}~(e)-(f), as solar generation increases, the net load is lower than the minimal conventional generation for a longer time. As a result, the usage of storage for energy arbitrage occurs during times in which the difference between the minimal and maximal demand is smaller, hence the profit from arbitrage is smaller.

% We also investigate the profitability of using energy storage for energy arbitrage. We find that as variable renewable energy generation increases, while the usage of storage for energy balancing increases, energy arbitrage becomes less profitable. We differentiate between two types of storage applications. When there is excess variable renewable energy generation, 
%i.e., when the net load is lower than the minimal power output of conventional power plants, 
% the grid operator can store it to balance supply and demand. We call this type of energy storage ``balancing storage''. Additional storage capacity, if available, can then be used for energy arbitrage to minimize electricity costs. We call this type of energy storage ``arbitrage storage''. 

% We find that as variable renewable energy generation increases, while ``balancing storage'' increases, the price the grid operator is willing to pay for ``arbitrage storage'' decreases. This can be seen in Fig.~\ref{fig:storage_usage}, which shows a heat map of the share of storage used for energy arbitrage (relative to the overall storage capacity) as a function of storage price bids, for three solar generation scenarios: 0\%, 5\%, and 10\% (relative to the daily demand). As solar generation increases, the price bids that the grid operator is willing to accept to use storage for energy arbitrage are lower. 

This can also be explained from a mathematical point of view by looking at the grid operator's cost function:
\begin{multline}
\frac{1}{2}\sum_{k = 0}^{K - 1} \left(\sum_{m = 1}^M p_{m,k}\right)^2  + \sum_{k = 0}^{K - 1} \sum_{m = 1}^M p_{L,k}p_{m,k} \\
+ \sum_{k = 0}^{K - 1} \sum_{m = 1}^M \frac{c_{m}}{a}p_{chg,m,k}.
\end{multline}

Increasing storage usage increases the quadratic component of the cost function, which increases the cost, and the linear component, which can \textit{decrease} the cost if the sum is negative (which is obtained by charging during low demand and discharging during peak demand). As solar penetration increases, the maximal negative linear component that can be obtained is smaller, thus the profitability of energy arbitrage is smaller as well. For similar reasons, storage arbitrage is more profitable in systems with more flexibility, i.e., lower minimal conventional generation output, since the difference between the low and peak demand is higher.

% To conclude, the higher the minimal conventional generation output is, storage is used more since it is required to absorb excess energy generated from solar, but any additional energy usage for arbitrage is less profitable. 

\begin{figure*}[htbp]
\centering
\includegraphics[width=0.95\textwidth]{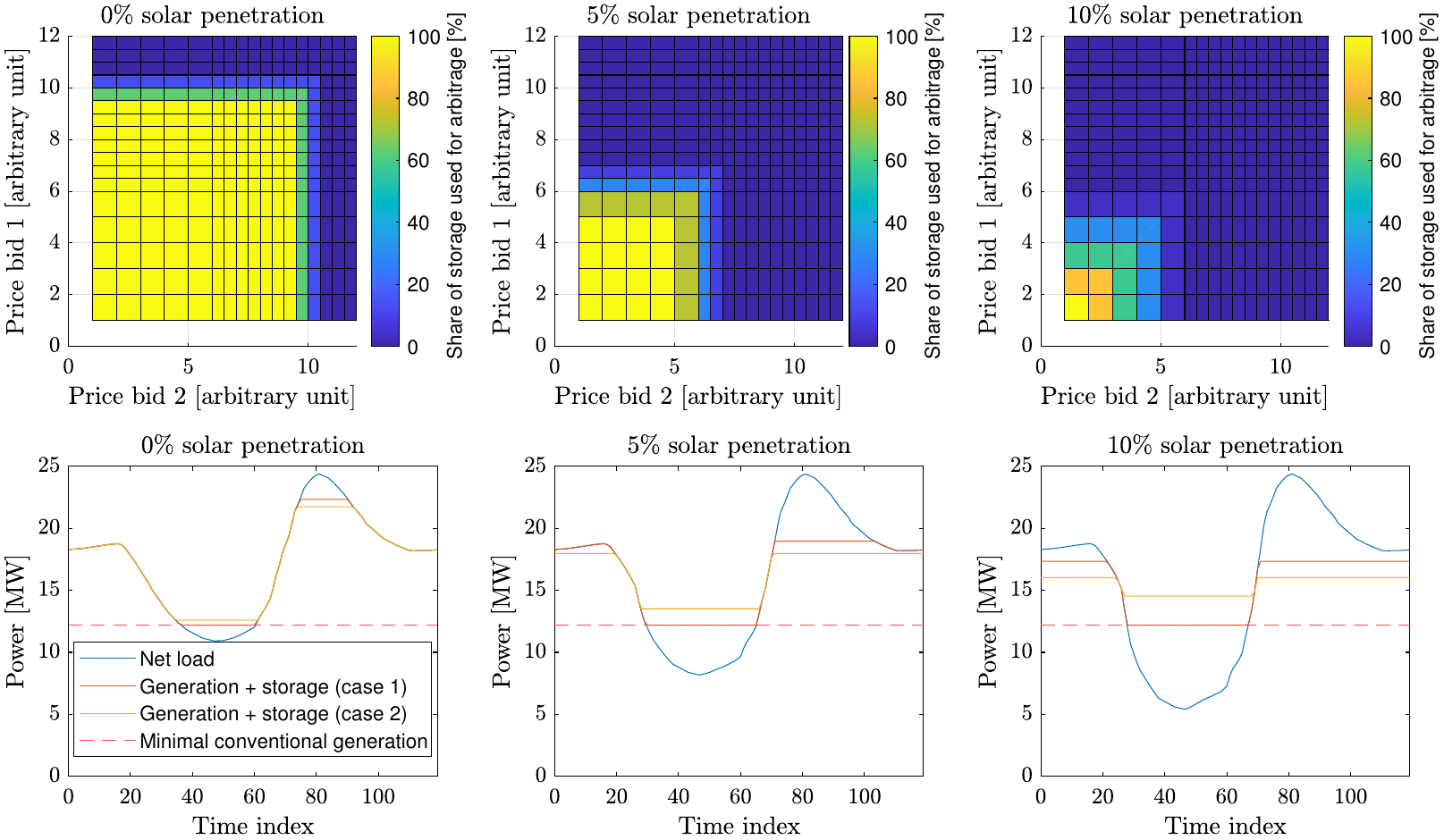}
\caption{Subplots (a)-(c) illustrate the percentage of storage allocated to energy arbitrage relative to the total available storage capacity at various storage price bids. The available storage capacity is calculated by deducting excess generated solar energy from the overall storage capacity. Subplots (d)-(f) display power profiles for each solar penetration scenario, presenting net load and conventional generation plus storage. Two cases are considered: (1) when the storage price bid is at its maximum (hence storage is used for balancing purposes only), and (2) when the storage price bid is at its minimum (reflecting storage usage for both balancing and arbitrage).}
\label{fig:storage_usage}
\end{figure*}

% We subtract the amount of energy that the grid operator is willing to purchase when the storage price bid is maximal from the amount of energy that it is willing to purchase when the bid is minimal. This gives us the amount of storage capacity that the grid operator is willing to pay for arbitrage purposes only.

\section{Conclusion}
\label{sec:conclusion}

In this paper, we study storage competition in wholesale electricity markets through a non-cooperative game between a grid operator and energy storage firms. We focus on the impact of the operational flexibility of conventional power plants on storage competition, an aspect that previous studies overlooked. A main result is that storage competition games are not necessarily stable, i.e., under certain conditions, the system will not converge to an equilibrium point. We show through both theory and simulation that instability is caused by a combination of high renewable energy generation, low flexibility in conventional power plants, and low storage capacity. However, we demonstrate that price stability can be obtained by using a price cap on storage price bids. We also find that as renewable energy generation increases, while the usage of storage for energy balancing purposes increases, the usage of energy storage for energy arbitrage decreases. 

An interesting future research direction could be to study the impact of storage reserve mechanisms on storage competition in wholesale electricity markets. Additionally, while we consider alternatives to storage, e.g., demand response and curtailment, as an external input that determines the maximal price bids that storage firms can offer, incorporating these alternatives into the analysis can provide insights regarding the dynamics of such interactions. % Moreover, a random selection mechanism between storage firms that offer the same price bids could decrease the expected utility of storage firms and thus ensure convergence to an equilibrium, without the need for a price cap. 
Finally, while we demonstrate that storage competition is unstable under certain conditions, studying the impact of low flexibility throughout the year to assess the annual impact on price stability and storage profit could be a valuable research effort. 

% While we show that under certain conditions the system will not converge to an equilibrium point, we do not prove that an equilibrium does not exist. Future research can attempt to prove that a pure strategy equilibrium does not exist, and explore the existence of mixed strategy Nash Equilibrium points. Moreover, while the proposed framework is suitable for the analysis of multiple energy storage firms, we only analyzed two market structures: a Monopoly and a Duopoly. It would be interesting to explore the stability of the game with more than two dominating energy storage firms.

% \section{Acknowledgment}\label{sec:Acknowledgment}

% \bibliographystyle{IEEEtran}
% \bibliography{database}

\vfill
\end{document}